\documentclass[10pt]{article}
\usepackage{a4wide}
\usepackage{latexsym,amsmath,amsfonts,amssymb,mathtools,mathrsfs}
\usepackage{array}
\usepackage{cite} %% To write [2,3,4,5] as [2-5]

\numberwithin{equation}{section}

\usepackage[unicode=true,colorlinks=true, urlcolor=blue, linkcolor=blue, citecolor=blue]{hyperref}

\DeclareMathOperator{\sign}{sgn}  % Sign of a real number/function
\newcommand{\dimM}{D}
\newcommand\dint[1]{#1\lrcorner}
\newcommand{\dex}{\mathrm{d}}
\newcommand{\Dex}{\boldsymbol{\mathrm{D}}}

\newcommand{\dform}[1]{\boldsymbol{#1}}        % Bold differential form

\newcommand{\dfal}{\dform{\alpha}}
\newcommand{\dfbe}{\dform{\beta}}
\newcommand{\dfga}{\dform{\gamma}}
\newcommand{\dfk}{\dform{k}}
\newcommand{\dfl}{\dform{l}}
\newcommand{\dfA}{\dform{A}}
\newcommand{\dfB}{\dform{B}}
\newcommand{\dfF}{\dform{F}}
\newcommand{\dfL}{\dform{L}}
\newcommand{\dfQ}{\dform{Q}}
\newcommand{\dfR}{\dform{R}}
\newcommand{\dfT}{\dform{T}}
\newcommand{\dfW}{\dform{W}}
\newcommand{\dfZ}{\dform{Z}}
\newcommand{\cofr}{\dform{\vartheta}}
\newcommand{\spatcofr}{\tilde{\cofr}}
\newcommand{\dfom}{\dform{\omega}}
\newcommand{\volf}{\dform{\omega}_\mathrm{vol}}

\newcommand{\teng}{\boldsymbol{g}}
\newcommand{\vpartial}{\boldsymbol{\partial}}
\newcommand{\vfre}{\boldsymbol{e}}

\newcommand{\LCtensor}{\mathcal{E}}

\newcommand{\kA}{\bar{A}}
\newcommand{\sA}{\tilde{A}}
\newcommand{\kB}{\bar{B}}
\newcommand{\sB}{\tilde{B}}
\newcommand{\kC}{\bar{C}}
\newcommand{\sC}{\tilde{C}}
\newcommand{\kP}{\bar{P}}
\newcommand{\sP}{\tilde{P}}

\newcommand{\calF}{\mathcal{F}}
\newcommand{\VdHC}{\tilde{V}}

\usepackage{amsthm}
\theoremstyle{plain}
    \ifx\thechapter\undefined
      \newtheorem{thm}{\protect\theoremname}   % thm is the name and the counter
    \else
      \newtheorem{thm}{\protect\theoremname}[chapter]
    \fi

    \newtheorem{prop}[thm]{Proposition}
    \newtheorem{corollary}[thm]{Corollary}
    \newtheorem{theorem}[thm]{Theorem}

\theoremstyle{definition}
    \newtheorem{defn}[thm]{Definition}

\begin{document}

\hfill{}\today

\vspace{.5cm}

%%%%%%%%%%%%%%%%%
\begin{center}
  {\Large \bf New metric-affine generalizations of \\  gravitational wave geometries}

  \vspace{.5cm}

    {\bf Alejandro Jim\'enez-Cano}
    \vspace{.2cm}

    {\it Departamento de F\'isica Te\'orica y del Cosmos and CAFPE \\
         Universidad de Granada, 18071, Granada, Spain
    } \\
    e-mail: {\tt alejandrojc@ugr.es}
\end{center}

\vspace{.2cm}

%%%%%%%%%%%%%%%%%
\begin{center} \bf ABSTRACT \end{center}
\vspace{.2cm}

\noindent

In this paper we explore generalizations of metric structures of the gravitational wave type to geometries containing an independent connection. The aim is simply to establish a new category of connections compatible, according to some criteria, to the known metric structures for gravitational waves and, additionally, provide some properties that can be useful for the search of solutions of this kind in different theories. 

%%%%%%%%%%%%%%%%%%%%%%%%%%%%%%%%%%%%%%%%%%%%%%%%%%%%%%%%%
\section{Introduction}

Metric-affine gravity is a natural extension of the geometry typically used in gravitational theories. In this framework the set of geometrical structures that participate in the gravitational physics contains, apart from the metric, a linear connection. This setting arises naturally when trying to formulate a gauge theory of certain groups of spacetime symmetries, such as the Poincar\'e or the Affine group. These two cases for example give rise to the so called Poincar\'e Gauge gravity (PG) and Metric Affine Gauge gravity (MAG) \cite{Blagojevic2001, Hehl1995}, respectively. In this formulation other properties of the matter fields, apart from the energy momentum tensor, are coupled to the geometrical structure. The spin density, dilation and shear currents enter the game as new fundamental properties of the matter related to the dynamics of the connection \cite{Hehl1976a,Hehl1976c,ObukhovTresguerres1993}. This formulation and its viability from the quantum gravity point of view has been getting some attention recently \cite{PercacciSezgin2019, Percacci2020}.

It is worth remarking that apart from MAG and PG gravity, many theories have been formulated considering an additional connection with certain properties or restrictions. See for example Ricci-Based Gravity theories containing the well-known $f(R)$ theories  \cite{AfonsoOlmoRubiera2018, AfonsoOlmoOraziRubiera2019, OlmoRubieraWojnar2019, Olmo2011, Koivisto2010, Pani2012, AfonsoOlmoOraziRubiera2018, BeltranDelhom2019,BeltranDelhom2020, BeltranHeisenbergOlmoRubiera2018, DelhomOlmoOrazi2019}, or the teleparallel equivalents and their generalizations \cite{BeltranHeisenbergKoivisto2018, AldrovandiPereira2012, BeltranDialektopoulos2020, KoivistoTsimperis2018, KrssakHoogenPereira2018,BeltranJimenezCano2019, BeltranHeisenbergKoivisto2019, HohmannJarvKrssakPfeifer2019}.

In purely metric (also called Riemannian) geometry, there is not a general covariant definition of what does it mean for a metric to represent a spacetime with gravitational radiation, for example in terms of certain property of its curvature. There are different criteria and conditions, as well as known geometries such as Kundt spacetimes, where we can find e.g. hypersurfaces playing the role of wave fronts \cite{Kundt1961, KundtTrumper2016}. Many of those criteria are collected in \cite{Zakharov1973} (see also \cite{Puetzfeld2000_Diploma} for a summary of some of them). Moreover, these criteria are constructed in the context of General Relativity, so when going to a generalization we have to ensure that the differential equations satisfied by the metric continue being compatible with the criteria. Regarding this aspect, we will not enter into details in this paper.

Obviously, when working with an independent linear connection, these definitions do not extrapolate due to the different nature of both geometrical structures. However, the already mentioned criteria for gravitational wave metrics are essentially conditions on the Riemann tensor associated to it or, to be more precise, to its Levi-Civita connection. Therefore they can be seen as a natural window to explore generalizations of these criteria, simply by considering their application to curvature tensors that come from other connections different from the Levi-Civita one. Indeed, if we consider a broader notion of curvature within a gauge theoretical context, the torsion, which appears as the fieldstrength of the translational part, can be regarded as a curvature as well and, hence, it can be subjected to these conditions. This way of extending the metric criteria is the idea we are going to explore in more detail throughout these pages, focusing on one particular criterion with an interesting physical meaning.

Finally, we would like to mention that exact gravitational wave solutions have been already explored in geometries including a non-Riemannian connection. See for example \cite{Obukhov2017, BlagojevicCvetkovic2017} for solutions in the context of PG, and also \cite{Obukhov2006, PasicVassiliev2005, Garcia2000, Vassiliev2005, Vassiliev2002} for solutions in metric-affine theories. In our study, we will also try to generalize the Ansatzes used by Obukhov in \cite{Obukhov2017, Obukhov2006} but, as we previously mentioned, respecting some criteria that can be obtained from the Riemannian ones by making a reasonable generalization.

We will start in Section \ref{sec: geom} with an overview about metric-affine geometry, and continue with a extensive compilation of results and criteria on gravitational wave geometries in metric gravity in Section \ref{sec: GW Rieman}. Then in Section \ref{sec: GW non-Rieman} we discuss how to extend the criteria of the previous section to the metric-affine framework. In Section \ref{sec: gen setting} we define a particular metric-affine geometry and see the conditions it needs to satisfy for the criteria to be fulfilled. A subcase of this geometry is considered in Section \ref{sec: subcase} and a simpler case of this one in Section \ref{sec: subsubcase}, where we show several useful properties. Finally, in Section \ref{sec: conclusion}, we provide a final discussion and some relevant remarks. At the end, our symbols are collected in Appendix \ref{app: symbols} and some other expressions in Appendix \ref{app: useful expr}.

\section{Review of fundamental objects in metric-affine geometry} \label{sec: geom}

Let $M$ be a $\dimM$-dimensional smooth manifold and $\vpartial_\mu\coloneqq \frac{\partial}{\partial x^\mu}$ a coordinate frame on it. The basic idea of metric-affine geometry is the inclusion of a new fundamental object in $M$, apart from the metric $\teng$, a \emph{linear connection} $\Gamma_{\mu\nu}{}^\rho$, which represents a notion of parallelism and allows to introduce a covariant derivative of tensors,
\begin{equation}
\nabla_\mu H^{\nu...}{}_{\rho...} \coloneqq \partial_\mu H^{\nu...}{}_{\rho...} 
+ \Gamma_{\mu\sigma}{}^\nu H^{\sigma...}{}_{\rho...} + \ldots
- \Gamma_{\mu\rho}{}^\sigma H^{\nu...}{}_{\sigma...} -\ldots
\end{equation}
and also defines a \emph{curvature tensor} and a \emph{torsion tensor} on $M$, respectively
\begin{align}
R_{\mu\nu\rho}{}^\lambda &\coloneqq \partial_\mu \Gamma_{\nu\rho}{}^\lambda -\partial_\nu \Gamma_{\mu\rho}{}^\lambda
+ \Gamma_{\mu\sigma}{}^\lambda \Gamma_{\nu\rho}{}^\sigma   - \Gamma_{\nu\sigma}{}^\lambda \Gamma_{\mu\rho}{}^\sigma\,,\\
T_{\mu\nu}{}^\rho &\coloneqq 2\Gamma_{[\mu\nu]}{}^\rho \,.
\end{align}

Once the metric is specified, an arbitrary linear connection is determined by its torsion and the so called \emph{non-metricity} tensor,
\begin{equation}
\qquad Q_{\mu\nu\rho} \coloneqq -\nabla_\mu g_{\nu\rho} \,,
\end{equation}
while for $T_{\mu\nu}{}^\rho=0$ and $Q_{\mu\nu\rho}=0$ we recover $\mathring{\Gamma}_{\mu\nu}{}^\rho$, the Levi-Civita connection of $\teng$.

In the context of gauge gravity or when we couple the geometry to matter fields (living in vector spaces under certain representations of the gauge group) it is specially useful to work in an arbitrary frame of the tangent bundle. Consider then a smooth distribution of basis in the tangent space of each point and its corresponding dual basis of 1-forms or \emph{coframe}\footnote{We will extensively use the vielbeins $ e^\mu{}_a$ and $e_\mu{}^a$ to change indices from one basis to the other. Examples:
\[
R_{\mu\nu a}{}^b =e^\rho{}_a e_\lambda{}^b R_{\mu\nu\rho}{}^\lambda,\quad 
T_{\mu\nu}{}^a = e_\lambda{}^a T_{\mu\nu}{}^\lambda,\quad 
Q_{\mu ab} =e^\nu{}_a e^\rho{}_b Q_{\mu\nu\rho},\quad 
\nabla_a=e^\mu{}_a\nabla_\mu, \quad \partial_a=e^\mu{}_a\partial_\mu\,.
\]
}
\begin{equation}
\vfre_a = e^\mu{}_a \vpartial_\mu\,\qquad \cofr^a = e_\mu{}^a \dex x^\mu\,,
\end{equation}
fulfilling  $e_\mu{}^ae^\mu{}_b=\delta_b^a$ (and $e^\mu{}_a e_\nu{}^a=\delta_\nu^\mu$). It is important to remark that this frame is in general anholonomic (i.e. no coordinate functions are associated) since the vectors of the basis have non-trivial Lie bracket with each other. This is reflected in the non-vanishing \emph{anholonomy coefficients},
\begin{equation}
\Omega_{ab}{}^c \coloneqq 2e^\mu{}_a e^\nu{}_b \partial_{[\mu}e^c{}_{\nu]}\quad  \Rightarrow \quad [\vfre_a,\,\vfre_b] = -\Omega_{ab}{}^c\vfre_c  \,.
\end{equation}
Actually, in terms of the dual basis these are indeed the coefficients of the exterior derivative of the coframe,
\begin{equation}
\dex\cofr^c = \tfrac{1}{2}\Omega_{ab}{}^c \cofr^a\wedge\cofr^b \,.
\end{equation}

The information regarding the linear connection is encoded in an object called the \emph{connection 1-form}, $\dfom_a{}^b = \omega_{\mu a}{}^b \dex x^\mu$, whose components are given by
\begin{equation}
\omega_{\mu a}{}^{b} \ = \  e^{\nu}{}_{a}\, e^{b}{}_{\lambda}\, \Gamma_{\mu\nu}{}^{\lambda} \ + \ e^{b}{}_{\sigma}\,\partial_{\mu} e^{\sigma}{}_{a}\,.
\end{equation}
reflecting the fact that the connection does not transform tensorially under basis transformations. From now on the operator $\nabla_\mu$ will act on both kind of indices, Greek and Latin, in one case with $\Gamma_{\mu\nu}{}^\rho$ and in the other with $\omega_{\mu a}{}^b$.

However when working with differential forms, the exterior derivative $\dex$ can also be covariantly extended as the \emph{exterior covariant derivative}. If we consider an arbitrary tensor-valued $p$-form
\begin{equation}
\dfal_{a...}{}^{b...} \ =\ \tfrac{1}{p!}\, \alpha_{\mu_1 ... \mu_p\,a...}{}^{b...} \, \dex x^{\mu_1} \wedge \ldots\wedge \dex x^{\mu_p}\,. \label{eq: arbitrary form}
\end{equation}
its exterior covariant derivative is defined as
\begin{equation}
    \Dex\dfal_{a...}{}^{b...}\ =\ \dex\dfal_{a...}{}^{b...} \ + \ \dfom_{c}{}^{b}\wedge\dfal_{a...}{}^{c...}\ + \ldots -\ \dfom_{a}{}^{c}\wedge\dfal_{c...}{}^{b...}\ -  \ldots\,,
    \label{extcovder}
  \end{equation} 
which for zero forms gives simply $\Dex\alpha_{a...}{}^{b...}=\nabla_\mu\alpha_{a...}{}^{b...} \dex x^\mu$. 
The curvature and torsion 2-forms and the non-metricity 1-form are then defined as 
  \begin{align}
    \dfR_{a}{}^{b} \, &  \coloneqq \ \dex\dfom_{a}{}^{b}\ + \ \dfom_{c}{}^{b}\wedge\dfom_{a}{}^{c}\,,\\
    \dfT^a      \,  &  \coloneqq \Dex\cofr^a  \,,\\ 
    \dfQ_{ab}      \, &  \coloneqq\ -\Dex g_{ab}\,.
  \end{align}
whose components (according to \eqref{eq: arbitrary form}) are $R_{\mu\nu a}{}^b$, $T_{\mu\nu}{}^a$ and $Q_{\mu ab}$, respectively. This forms (or, equivalently, the corresponding tensors) can be decomposed according to the irreducible representations of $\mathrm{GL}(\dimM,\mathbb{R})$. For the torsion we have the trace part $\dfT^{(\mathrm{tr})}$, the totally antisymmetric part $\dfT^{(\mathrm{a})}$ and the rest $\dfT^{(\mathrm{tn})}$; for the non-metricity we have both traces $\dfQ^{(\mathrm{tr}1)}$ and $\dfQ^{(\mathrm{tr}2)}$, the traceless totally symmetric part $\dfQ^{(\mathrm{s})}$ and the rest $\dfQ^{(\mathrm{tn})}$; and finally, the curvature, in the presence of a metric, can be first separated into symmetric and antisymmetric parts,
\begin{equation}
\dfR_{ab} =  \dfZ_{ab} +\dfW_{ab}\,,
\end{equation}
where $\dfZ_{ab}\coloneqq\dfR_{(ab)}$ and $\dfW_{ab}\coloneqq\dfR_{[ab]}$. Then, it can be shown that $\dfW_{ab}$ splits into six irreducible parts $\dfW^{(A)}{}_{ab}$ ($A=1,...,6$) and $\dfZ_{ab}$ into five $\dfZ^{(A)}{}_{ab}$ ($A=1,...,5$). For the detailed expressions of all of these irreducible components in arbitrary dimensions, see \cite{Hehl1995, McCrea1992}.\footnote{The dictionary from our notation to the one they use is:
\[
\dfT^{(\mathrm{tr})}\rightarrow{}^{(2)}T\,,\qquad \dfT^{(\mathrm{tn})}\rightarrow{}^{(1)}T\,,\qquad \dfT^{(\mathrm{a})}\rightarrow{}^{(3)}T\,,\qquad \dfZ^{(i)}\rightarrow{}^{(i)}Z\,,\qquad \dfW^{(i)}\rightarrow{}^{(i)}W\,,
\]
\[
\dfQ^{(\mathrm{tr}1)}\rightarrow{}^{(4)}Q\,,\qquad \dfQ^{(\mathrm{tr}2)}\rightarrow{}^{(3)}Q\,,\qquad \dfQ^{(\mathrm{tn})}\rightarrow{}^{(2)}Q\,,\qquad \dfQ^{(\mathrm{s})}\rightarrow{}^{(1)}Q\,.
\]}

\section{Gravitational waves in Riemannian geometry} \label{sec: GW Rieman}

\subsection{Transversal space}
Given a lightlike vector field $k^{\mu}$, in order to define the transverse space we need to introduce another lightlike vector $l^{\mu}l_{\mu}=0$, such that $k^{\mu}l_{\mu}\neq0$. Since the normalization for a null vector is arbitrary, let us consider without loss of generality that the field $l^{\mu}$ verifies
\begin{equation}
k^{\mu}l_{\mu}=1\,.
\end{equation}

\begin{defn}
\textbf{(Transversal)}. Given a lightlike congruence with velocity $k^{\mu}$ and another non-colinear lightlike vector $l^{\mu}$ that satisfy $l_{\mu}k^{\mu}=1$, the orthogonal 
\begin{equation}
\left(\mathrm{span}_{\mathbb{R}}\{k^{\mu}\vpartial_{\mu},\,l^{\mu}\vpartial_{\mu}\}\right)^{\bot}\,,
\end{equation}
is called \emph{transversal space} of the congruence (with respect
to $l^{\mu}$). We will say that a tensor $H^{\mu_{1}...\mu_{r}}{}_{\nu_{1}...\nu_{s}}$
is \emph{transversal} if the contraction of any of its indices with
$l^{\mu}$ and $k^{\mu}$ vanishes.
\end{defn}

At this point it is useful to introduce the projector onto the transverse spatial slices (see a more detailed explanation in \cite{Blau2011}),
\begin{equation}
h^\mu{}_\nu \coloneqq \delta^\mu_\nu-k^\mu l_\nu-l^\mu k_\nu\,,
\end{equation}
that allows us to extract the transversal part of any tensor, which we will denote by a tilde over it,
\begin{equation}
\tilde{H}{}^{\mu \nu ...}{}_{\rho \lambda ...}\coloneqq h^\mu{}_\alpha h^\nu{}_\beta \cdots \ h^\gamma{}_\rho h^\delta{}_\lambda\cdots \  H^{\alpha \beta ...}{}_{\gamma \delta ...}\,.
\end{equation}

\subsection{Null congruences and optical scalars}

Let us now present some quantities that characterize the behavior of a given lightlike congruence with velocity $k^\mu$. Consider the tensor
\begin{equation}
B^{\nu}{}_{\mu}\coloneqq\mathring{\nabla}_{\mu}k^{\nu}\,.\label{eq: def B tensor}
\end{equation}
Its transversal part can be decomposed as\footnote{We absorb the factor $(\dimM-2)^{-2}$ of the trace part into the definition of $\theta$, as it is usual in the literature.}
\begin{equation}
\tilde{B}{}_{\mu\nu} = \omega_{\mu\nu}+\sigma_{\mu\nu}+h_{\mu\nu}\theta\,,
\end{equation}
where we have introduced
\begin{align}
\omega_{\mu\nu} & \coloneqq \tilde{B}{}_{[\mu\nu]}\,,\\
\theta & \coloneqq\tfrac{1}{\dimM-2}h^{\mu\nu}\tilde{B}{}_{(\mu\nu)}\,,\\
\sigma_{\mu\nu} & \coloneqq \tilde{B}{}_{(\mu\nu)}-h_{\mu\nu}\theta \,.
\end{align}
called, respectively, the \emph{twist tensor}, the \emph{expansion} scalar and the \emph{shear tensor}, and whose expressions in terms of $k^\mu$, $l^\nu$ and the projector $h^\mu{}_\nu$ are collected in the Appendix \ref{app:opticaldecom}. For a given $l^\mu$, this decomposition is unique. Making use of these objects one can
  construct\footnote{ Note that the quantities $\omega_{\mu\nu}\omega^{\mu\nu}$ and $\sigma_{\mu\nu}\sigma^{\mu\nu}$ are non-negative due to the transversality of $\omega_{\mu\nu}$ and $\sigma_{\mu\nu}$.}
\begin{align}
\omega & \coloneqq\sqrt{\tfrac{1}{\dimM-2}\omega_{\mu\nu}\omega^{\mu\nu}}\,, \\
|\sigma| & \coloneqq\sqrt{\tfrac{1}{\dimM-2}\sigma_{\mu\nu}\sigma^{\mu\nu}}\,. 
\end{align}
The objects $\left\{\theta,\,\omega,\,|\sigma|\right\}$ (expansion, twist and shear) are known as the \emph{optical scalars} of the congruence. 

\begin{defn}
\textbf{(Normal congruence)}. A congruence is \emph{normal} if there
exists a family of hypersurfaces orthogonal to the curves it contains.
\end{defn}

A very interesting result is the following \cite{BekaertMorand2013}:

\begin{prop}
\label{Prop: normal geod} In a semi-Riemannian manifold $(M,\,\teng)$,
any normal lightlike congruence is pre-geodetic. Therefore it can
be reparameterized to get a geodetic congruence.
\end{prop} 

Due to this result, we concentrate on geodetic congruences for which the optical scalars are given by
\begin{align}
\theta & =\tfrac{1}{\dimM-2}\mathring{\nabla}_{\sigma}k^{\sigma}\,,\\
\omega^2 & =\tfrac{1}{\dimM-2}\partial_{[\mu}k_{\nu]}\partial{}^{\mu}k^{\nu}\,,\\
|\sigma|^2 & =\tfrac{1}{\dimM-2}\mathring{\nabla}_{(\mu}k_{\nu)}\mathring{\nabla}^{\mu}k^{\nu}-\theta^{2}\,.
\end{align}
which can be expressed in exterior notation
  as\footnote{The symbol $\dint{\boldsymbol{V}}$ represents the interior product by a vector.}$^{,}$\footnote{For the Hodge star operator acting on a $p$-form we use the following convention,
  \begin{equation}
  \star\dfal \ = \ \tfrac{1}{(\dimM-p)!p!}\, \alpha^{b_1...b_p} \, \LCtensor_{b_1...b_p c_1...c_{\dimM-p}}\, \cofr^{c_1}\wedge ...\wedge\cofr^{c_{\dimM-p}}\,.
  \end{equation}
  in terms of the \emph{Levi-Civita tensor}, $\LCtensor_{a_1 \dots a_\dimM} \coloneqq \sqrt{|\det(g_{ab})|} \ \dimM! \delta^1_{[a_1} ... \delta^\dimM_{a_\dimM]}$.}
\begin{align}
\theta & =\tfrac{1}{\dimM-2}\sign(g)\star\dex\star\dfk\,,\\
\omega^{2} & =\tfrac{1}{2(\dimM-2)}\sign(g)\star\left(\dex\dfk\wedge\star\dex\dfk\right)\,,\\
|\sigma|^{2} & =\tfrac{1}{\dimM-2}\left[\dint{\vfre_{(a}}\mathring{\Dex}\left(\dint{\vfre_{b)}}\dfk\right)\right]\left[\dint{\vfre^{a}}\mathring{\Dex}\left(\dint{\vfre^{b}}\dfk\right)\right]-\theta^{2}\,.
\end{align}
 
We end this subsection on optical scalars by remarking a very useful property of lightlike congruences that relates the nullity of the twist with the existence of wave fronts (see \cite{Kundt1961, KundtTrumper2016,Poisson2002} and \cite[p.~59]{Witten1962}):

\begin{prop}
\label{Prop: normal null twist zero} A lightlike geodetic congruence is normal if and only if the twist $\omega$ vanishes.
\end{prop}

More information on twist-free solutions of pure radiation can be found in \cite{KundtTrumper2016}.

\subsection{Kundt and Brinkmann metrics}

Now we focus our interest in the concept of plane-fronted waves. Kundt defined them in \cite{Kundt1961}, a definition that was also presented in \cite[p.~85--86]{Witten1962} together with a theorem that introduces a characterization: \emph{a plane-fronted wave is a vacuum field that admits a normal null congruence with $|\sigma|=\theta=0$}. Observe that, as a consequence of Propositions \ref{Prop: normal geod} and \ref{Prop: normal null twist zero}, the congruence of a plane-fronted wave is pre-geodetic (so it can be expressed as geodetic changing the velocity appropriately) and, additionally, $\omega=0$. 

These definitions correspond to ``vacuum solutions'' of General Relativity, which might not be solutions for other more general theories. Since we are interested in spacetimes defined in a theory-independent way, we start introducing:\footnote{Note that no allusion to vacuum has been made.}

\begin{defn}
\textbf{(Kundt space)}. A \emph{Kundt space}  is a Lorentzian manifold
that admits a geodetic null congruence with $|\sigma|=\theta=\omega=0$.
\end{defn}

Every point in a Kundt space admits a coordinate chart $\{x^{\mu}\}=\{u,\,v,\,z^{2},...,z^{\dimM-1}\}$
(a \emph{Kundt chart}) in which the line element is expressed:
\begin{equation}
\mathrm{d} s^{2}=2\dex u\dex v+H(u,\,v,\,z)\dex u^{2}+2W_{i}(u,\,v,\,z)\dex u\dex z^{i}+\tilde{g}_{ij}(u,\,z)\dex z^{i}\dex z^{j}\,,
\end{equation}
where $i,j=2,\,...,\,\dimM-1$ and $\tilde{g}_{ij}$ is the spatial metric with signature $(-,...,-)$. We have then a local foliation by spacelike surfaces, those with constant $u$. The expressions for the Christoffel symbols, Riemann and Ricci tensors can be found in \cite{Podolsky2009} or \cite[p.~230--231]{BicakLedvinka2014}. Solutions of this kind in different backgrounds (e.g. with and without cosmological constant) are given in \cite[chap.~18]{GriffithsPodolsky2009} and \cite[chap.~31]{Stephani2003}. 

Observe that the coordinate field $\vpartial_{v}\eqqcolon k^\mu \vpartial_{\mu}$, which satisfies
\begin{equation}
k_{\mu}k^{\mu}=0\,,\qquad\qquad k^\rho\mathring{\nabla}_{\rho}k^{\mu}=0\,, \label{eq:wavevector}
\end{equation}
is indeed the velocity field of the congruence that appears in the definition. Moreover, note that $\vpartial_{v}$ is not a covariantly constant field with respect to $\mathring{\nabla}$,
\begin{equation}
\mathring{\nabla}_{\mu}k^\rho=\mathring{\Gamma}_{\mu v}{}^{\rho} =\tfrac{1}{2}\left(g^{\rho u}\partial_{v}g_{\mu u}+g^{\rho i}\partial_{v}g_{\mu i}\right)\neq 0\,,
\end{equation}
i.e. the tensor $B_{\mu\nu}$ defined in \eqref{eq: def B tensor} is not trivial for Kundt spaces. This expression vanishes if $H$ and $W_{i}$ are independent of the coordinate $v$. This is a well known particular kind of Kundt spaces called Brinkmann spaces \cite{Brinkmann1925}:

\begin{defn}
\textbf{(Brinkmann space)}. A \emph{Brinkmann space}  is a lorentzian manifold that admits a non-vanishing vector field $k^\mu \vpartial_\mu$ which is lightlike and covariantly constant with respect to the Levi-Civita conection, namely
\begin{equation}
k_\mu k^\mu=0\,\quad\text {and}\quad\mathring{\nabla}_{\rho}k^{\mu}=0\,.
\end{equation}
\end{defn}

If we introduce the associated 1-form $\dfk=k_{\mu}\dex x^{\mu}$, this two conditions can be written in the exterior notation, respectively,
\begin{equation}
\dfk\wedge\star\dfk=0\,,\qquad\qquad\mathring{\Dex}k^a =0\,.
\end{equation}
In an analogous way as in the Kundt case, there is a local chart we can always find, $\{x^{\mu}\}=\{u,\,v,\,z^{2},...,z^{\dimM-1}\}$ (a \emph{Brinkmann chart}), that allows to express the metric (see for example \cite{Blanco2011}):
\begin{equation}
\mathrm{d} s^{2}=2\dex u\dex v+H(u,\,z)\dex u^{2}+2W_{i}(u,\,z)\dex u\dex z^{i}+\tilde{g}_{ij}(u,\,z)\dex z^{i}\dex z^{j}\,.\label{eq: Brinkmann general metric}
\end{equation}
Moreover, $H$ or $W_{i}$ (but not both) can always be set to zero with an appropriate redefinition of the spatial coordinates $\{z^{i}\}$ (see for example \cite{Ortin2004}). From now on, when we refer to the Brinkmann metric we will take $W_{i}=0$, so the metric becomes block-diagonal.

\begin{defn}
\textbf{\label{def: ppwave}(pp-wave)}. A \emph{plane-fronted wave with parallel rays }(or\emph{ pp-wave}) is a Brinkmann space admitting a coordinate chart in which the metric is expressed
\begin{equation}
\mathrm{d} s^{2}=2\dex u\dex v+H(u,\,z)\dex u^{2}-\delta_{ij}\dex z^{i}\dex z^{j}\,,\label{eq: ppwave}
\end{equation}
\end{defn}

By calculating the Einstein tensor of \eqref{eq: ppwave} it is straightforward to prove that this is a vacuum solution of the Einstein equations if and only if $H$ is a harmonic function of the transversal coordinates,
\begin{equation}
\partial_{i}\partial^{i}H(u,\,z)=0\,.
\end{equation}
This condition obviously will no longer be true in more general theories.

\subsection{Criteria for gravitational wave spacetimes}

As we have already mentioned, there are many different attempts in the literature trying to (covariantly) characterize spacetimes with gravitational radiation in General Relativity. These approaches are based on a previous analysis of the Einstein equations and the existence of characteristic submanifolds (wave fronts) and bicharacteristics (rays). Several of these criteria are extensively studied in \cite{Zakharov1973}. Based on that reference and the overview in \cite{Puetzfeld2000_Diploma}, here we present some of them:

\begin{itemize}

\item \noindent \textbf{Pirani criterion}. We will say there are free gravitational waves in an empty region of a spacetime if and only if the curvature there is type $\mathbf{II}$, $\mathbf{III}$ or $\mathbf{N}$ in the Petrov classification.

\item \noindent \textbf{Lichnerowicz criterion}. For a non-vanishing curvature $\mathring{\dfR}_{a}{}^{b}\neq0$, we will say there is gravitational radiation if and only if there exists a non-vanishing 1-form $\dfk=k_{\mu}\dex x^{\mu}$ satisfying
\begin{align}
\dfk\wedge\star\mathring{\dfR}_{a}{}^{b} & =0 &  & \Leftrightarrow & k^{\mu}\mathring{R}_{\mu\nu a}{}^{b} & =0\,,\label{eq: LichneC1}\\
\dfk\wedge\mathring{\dfR}_{a}{}^{b} & =0 &  & \Leftrightarrow & k_{[\mu}\mathring{R}_{\nu\rho]a}{}^{b} & =0\,.\label{eq: LichneC2}
\end{align}
Lichnerowicz proved that these two conditions, under the hypothesis $\mathring{\dfR}_{a}{}^{b}\neq0$, imply that $k^{\mu}$ is both lightlike and geodetic \cite{Zakharov1973}. For example, the lightlike condition is immediate contracting \eqref{eq: LichneC2} with $k^{\mu}$ and then substituting \eqref{eq: LichneC1}. In addition, another consequence is that the curvature can be written:
\begin{equation}
\mathring{R}_{\mu\nu\rho\lambda}=b_{\mu\rho}k_{\nu}k_{\lambda}+b_{\nu\lambda}k_{\mu}k_{\rho}-b_{\mu\lambda}k_{\nu}k_{\rho}-b_{\nu\rho}k_{\mu}k_{\lambda}\quad\Leftrightarrow\quad \mathring{R}_{\mu\nu}{}^{\rho\lambda}=4b_{[\mu}{}^{[\rho}k_{\nu]}k^{\lambda]}\,,
\end{equation}
for some symmetric tensor $b_{\mu\nu}=b_{(\mu\nu)}$ with the property
$k^{\mu}b_{\mu\nu}=0$. 

\item \noindent \textbf{Zel'manov criterion}. We will say there is gravitational radiation in a spacetime region if and only if the curvature of this region is not covariantly constant, i.e. $\mathring{\nabla}_{\sigma}\mathring{R}_{\mu\nu a}{}^{b}\neq0$, and verifies the following covariant generalization of the wave equation
\begin{equation}
\mathring{\nabla}^{\sigma}\mathring{\nabla}_{\sigma}\mathring{R}_{\mu\nu a}{}^{b}=0\,.
\end{equation}

This condition is formulated in terms of a particular connection (Levi-Civita). Another criterion very similar to this one but formulated independently of any connection is the Maldybaeva criterion, which is based on a special (metric dependent) operator that acts on differential forms:

\item \noindent \textbf{Maldybaeva criterion}. We will say there is gravitational radiation in a spacetime region if and only if the (non-trivial) curvature 2-form satisfies the wave equation
\begin{equation}
\Delta\mathring{\dfR}_{a}{}^{b}=0\,,
\end{equation}
where $\Delta\coloneqq\dex\delta+\delta\dex$ is the Laplace-de Rham operator.

\end{itemize}

In the context of General Relativity one can easily find the following relations between these criteria for the particular case of Einstein spaces:

\begin{prop}
\textbf{\textup{\label{prop: relation criteria}Relations between criteria for Einstein spaces.}}

Let $(M,\teng)$ be an Einstein space, namely one satisfying $R_{\mu\rho\nu}{}^\rho=cg_{\mu\nu}$ for some real constant $c$. Then, the following statements hold:
\begin{itemize}
\item In vacuum $(c=0)$: Lichnerowicz criterion $\Leftrightarrow$ Petrov type $\mathbf{N}$.
\item Maldybaeva criterion $\Leftrightarrow$ vacuum $(c=0)$ and Petrov type $\mathbf{N}$.
\item Zel'manov criterion $\Rightarrow$ vacuum $(c=0)$ and Petrov type $\mathbf{N}$. \\
The converse (vacuum $+\mathbf{N}\Rightarrow$ Zel'manov) is also true with only the metrics\textup{ \cite[eq. (7.12)]{Zakharov1973}} as exceptions.
\end{itemize}
\end{prop}
\noindent The result is the following diagram for the criteria we have seen in the case of Einstein spaces:
\begin{equation}
\begin{array}{ccccc}
\text{Zel'manov} & \begin{array}{c}
\Longrightarrow\\
\Longleftarrow_{*}
\end{array} & \text{Pirani (N)}+\text{vacuum} & \Leftrightarrow & \text{Maldybaeva}\\
 &  & \Updownarrow\\
 &  & \text{Lichnerowicz}\\
 &  & +\text{vacuum}
\end{array}
\end{equation}
where $*$ denotes that there are two exceptions. For non-Einstein spaces, the relations become more obscure. 

Finally, it is worth mentioning that there are other criteria, such as the ones by Debever, Bel, etc. More information about them can be found in \cite{Zakharov1973}. 

\section{Extension to metric-affine geometries} \label{sec: GW non-Rieman}

In the previous section we have seen that there are different criteria to classify or categorize metrics in a gravitational wave type. However, in a metric-affine framework we have an additional field, a linear connection $\Gamma_{\mu\nu}{}^\rho$ (or equivalently $\omega_{\mu\nu}{}^a$), and the idea now is to analyze how can we restrict in a reasonable way an arbitrary connection to explore gravitational wave scenarios in these theories.

We start by recalling that Petrov types are based on the classification of the principal null directions of the Weyl tensor. Therefore, a generalization of the Pirani criterion could be possible by understanding the behavior of the irreducible part $\dfW^{(1)}{}_{ab}$ of the new curvature, which is the one that reduces to the Weyl tensor in a Riemannian geometry \cite{Hehl1995}.

For the Maldybaeva and Zel'manov criteria one could generalize the differential operator, the curvature or both. Actually, when working either in MAG or in PG gravity, the fields playing the role of gauge fieldstrengths (gauge curvatures) are the curvature $\dfR_a{}^b$ and the torsion $\dfT^a$ of the connection, so another possibility might be to apply the criterion to both objects.

Finally, we examine the Lichnerowicz criterion for gravitational waves. If one checks in detail electromagnetic wave configurations in classical Maxwell theory, it is easy to see that the curvature form  $\dfF= \tfrac{1}{2}F_{\mu\nu} \dex x^\mu\wedge\dex x^\nu$ associated to $A_\mu$ satisfies the analogous conditions
\begin{equation}
\dfk\wedge\dfF=0\,,\qquad \dfk\wedge\star\dfF=0\,.
\end{equation}
If we now look at these equations in components,
\begin{equation}
k_{[\mu}F_{\nu\rho]}=0\,,\qquad  k^{\mu}F_{\mu\nu}=0\,,
\end{equation}
we realize that they essentially encode the well-known radiation conditions for the electromagnetic field,
\begin{equation}
\delta_{ij}k^{i}E^{j}=\delta_{ij}k^{i}B^{j}=0\,,\qquad\epsilon_{ijk}\frac{k^{i}}{k^{0}}E^{j}=\delta_{ik}B^{i}\,.
\end{equation}
After this motivation, and inspired
  by\footnote{Although, in \cite{BlagojevicCvetkovic2017} the authors use a different generalization: they contract $k$ with the internal indices:
  \[
  k^a R_{\mu\nu ab}=0\,,\qquad R_{\mu\nu[ab}k_{c]}=0\,,\qquad k^a T_{\mu\nu a}=0\,,\qquad T_{\mu\nu[a}k_{c]}=0\,.
  \]}
\cite{Obukhov2017} and \cite{BlagojevicCvetkovic2017} we are going to focus on the Lichnerowicz criterion and its generalization to a metric-affine geometry imposing the corresponding conditions over our curvatures. Considering we are working in a MAG framework, we define the following generalization involving both fieldstrengths:
\begin{defn} {\bf (Generalized) Lichnerowicz Criteria \label{def: GenLC}}
\begin{align}
\text{1LCR}&&  \dfk\wedge\dfR_{a}{}^{b} & =0 &  & \Leftrightarrow & k_{[\mu}R_{\nu\rho]a}{}^{b} & =0\,,\nonumber \\
\text{2LCR}&& \dfk\wedge\star\dfR_{a}{}^{b} & =0 &  & \Leftrightarrow & k^{\mu}R_{\mu\nu a}{}^{b} & =0\,,\label{eq: General Lich R}\\
\text{1LCT}&& \dfk\wedge\dfT{}^{a} & =0 &  & \Leftrightarrow & k_{[\mu}T_{\nu\rho]}{}^{a} & =0\,,\nonumber \\
\text{2LCT}&& \dfk\wedge\star\dfT{}^{a} & =0 &  & \Leftrightarrow & k^{\mu}T_{\mu\nu}{}^{a} & =0\,.\label{eq: General Lich T}
\end{align}
\end{defn}

Now we provide an important general result which allows to express the curvature and the torsion form under these conditions in a very special way.

\begin{prop}
\label{prop: Lich crit and quad general}Consider a lightlike congruence with velocity $k^\mu$. Let $\dfk=k_\mu \dex x^\mu$ be the associated 1-form and $\dfl$ another lightlike 1-form such that $l^{\mu}k_{\mu}=1$. For an arbitrary tensor-valued 2-form $\dfal_{a...}{}^{b...}$, the following results hold:
\begin{enumerate}
\item If  $\dfk\wedge\dfal_{a...}{}^{b...}=0$, $\dfal$ can be expressed
\begin{equation}
\dfal_{a...}{}^{b...}=\dfk\wedge(s_{a...}{}^{b...}\, \dfl +\tilde{\dfbe}_{a...}{}^{b...})\,,
\end{equation}
where $s_{a...}{}^{b...}$ is a tensor-valued 0-form and $\tilde{\dfbe}_{a...}{}^{b...}$ is a tensor-valued transversal 1-form.\footnote{From now on \textit{transversal} means transversal to the congruence generated by $\dfk$ with respect to $\dfl$.}

\item If  $\dfk\wedge\star\dfal_{a...}{}^{b...}=0$, then
\begin{equation}
\dfal_{a...}{}^{b...}=\dfk\wedge\tilde{\dfbe}_{a...}{}^{b...}+\tilde{\dfga}_{a...}{}^{b...} \,,
\end{equation}
for certain tensor-valued transversal forms $\tilde{\dfbe}_{a...}{}^{b...}$ and $\tilde{\dfga}_{a...}{}^{b...}$.

\item If the two conditions of the previous points are fulfilled, the 2-form reduces to
\begin{equation}
\dfal_{a...}{}^{b...}=\dfk\wedge\tilde{\dfbe}_{a...}{}^{b...}\,,
\end{equation}
and the following quadratic condition is
  satisfied,\footnote{Applied to electromagnetic waves in Maxwell theory, this quadratic condition lead us to $\dfF\wedge\star\dfF=0$, namely $F_{\mu\nu} F^{\mu\nu}=0$, that corresponds to the equality $  \delta_{ij}E^iE^j =\delta_{ij}B^iB^j$.}
\begin{equation}
\dfal_{a...}{}^{b...}\wedge\star\dfal_{c...}{}^{d...}=0\, \qquad\big(\alpha_{\mu\nu a...}{}^{b...} \ \alpha^{\mu\nu}{}_{c...}{}^{d...}=0\big) \,.
\end{equation}
\end{enumerate}
\end{prop}

\begin{proof}
We drop the external indices for simplicity. An arbitrary two form can be decomposed:
\[\dfal= s \ \dfk\wedge\dfl+ \dfk\wedge\tilde{\dfbe}+\dfl\wedge\tilde{\dfbe}' +\tilde{\dfga}\,.\]
where $\tilde{\dfbe}$ and $\tilde{\dfbe}'$ are transversal 1-forms and $\tilde{\dfga}$ is a transversal 2-form. 
\begin{enumerate}
\item The condition tells
\begin{equation}
0  =\dfk\wedge\dfal =\dfk\wedge\dfl\wedge\tilde{\dfbe}'+\dfk\wedge\tilde{\dfga}\,,
\end{equation}
so by linear independence, $\tilde{\dfbe}'=\tilde{\dfga}=0$.
\item Here we use that $\dfk\wedge\star$ acts on differential forms as the operator $\star\dint{(k^\mu \vpartial_\mu)}=\star\dint{\vpartial_v}$ up to a constant. So no components in the direction of $\dex v$ (equivalently, $\dfl$) are allowed, as can be seen by doing the calculation,
\begin{equation}
0 = \dfk\wedge\star\dfal=-\star(-s\ \dfk+\tilde{\dfbe}')\,.
\end{equation}
And this is true if and only if  $s=\tilde{\dfbe}'=0$.
\item They are immediate consequences of the previous results.
\end{enumerate}
\end{proof}

\section{General metric-affine setting} \label{sec: gen setting}

\subsection{Metric structure}

Due to its simplicity, we are interested in a spacetime metric of the Brinkmann type,\footnote{From now on we choose the orientation $\mathcal{E}_{u,v,z^2,...,z^{\dimM-1}}= \sqrt{|g|}$ }
\begin{equation}
\mathrm{d} s^{2}=2\dex u\dex v+H(u,\,z)\dex u^{2}+\tilde{g}_{ij}(u,\,z)\dex z^{i}\dex z^{j}\,.\label{eq: Brink metric}
\end{equation}
As we have seen previously, from this metric we can obtain two relevant (dual) objects. The first one is the \emph{wave vector} $k^{\mu}\vpartial_{\mu}\coloneqq \vpartial_{v}$, that points towards the direction of propagation of the wave. It is autoparallel with respect to $\mathring{\nabla}$ and lightlike by definition of the $v$ coordinate. The other one is the \emph{wave form} $\dfk\coloneqq k_\mu\dex x^{\mu}$, which is indeed $\dex u$ and, consequently, an exact form.

In addition, consider the 1-form 
\begin{equation}
\dfl=l_{\mu}\dex x^{\mu}\coloneqq\tfrac{1}{2}H(u,\,z)\dex u+\dex v\,.
\end{equation} 
It is not difficult to see that it is lightlike and verifies $l_{\mu}k^{\mu}=1$. Clearly, the transversal space of the congruence generated by $k^\mu$ with respect to $l^\mu$ is the one generated by the coordinate vectors in $\{z^{i}\}$ directions,
\begin{equation}
\tilde{T}_{p}M\coloneqq\mathrm{span}_{\mathbb{R}} \left\{ \vpartial_{i}|_{p} \right\}_{i=2}^{\dimM-1} \ \subset \ T_{p}M \,.
\end{equation}

\subsection{Coframe}

In a theory with  $\mathrm{GL}(\dimM,\,\mathbb{R})$ freedom, we can choose a \emph{lightcone gauge} for the first two directions and an \emph{orthonormal gauge} for the transversal space. Then
\begin{equation}
\teng\coloneqq\cofr^{0}\otimes\cofr^{1}+\cofr^{1}\otimes\cofr^{0}-\delta_{IJ}\cofr^{I}\otimes\cofr^{J}\quad=g_{ab}\cofr^a \otimes\cofr^b \,, \label{eq: anhol metric}
\end{equation}
There are several coframes compatible with this gauge. But one special coframe (we will call \emph{gauge basis}) that makes this possible for the Ansatz we have taken for the metric \eqref{eq: Brink metric} is:
\begin{equation}
\{\cofr^a \}=\left\{ \begin{array}{rlr}
\cofr^{0} & \coloneqq\dfk=\dex u \,,\\
\cofr^{1} & \coloneqq\dfl=\frac{1}{2}H\dex u+\dex v\,,\\
\cofr^{I} & \coloneqq e_{i}{}^{I}\dex z^{i}\,
\end{array}\right.
\end{equation}
with dual frame
\begin{equation}
\{\vfre_{a}\}=\left\{ \begin{array}{rlr}
\vfre_{0} & =-\frac{1}{2}H\vpartial_{v}+\vpartial_{u}\,,\\
\vfre_{1} & =\vpartial_{v}\,,\\
\vfre_{I} & =e^{i}{}_{I}\vpartial_{i}\,
\end{array}\right.
\end{equation}
where the vielbeins $e_{i}{}^{I}(u,\,z)$ satisfy
\begin{equation}
-\delta_{IJ}e_{i}{}^{I}e_{j}{}^{J}=\tilde{g}_{ij}\,\qquad
e^{i}{}_{J}e_{i}{}^{I}=\delta_{J}^{I}\,\qquad e^{i}{}_{I}e_{j}{}^{I}=\delta_{j}^{i}\,.
\end{equation}
In order to work always with $\dimM$-dimensional indices and avoid using the index $I$, we define the transversal coframe,
\begin{equation}
\spatcofr^{a}\coloneqq\delta_{I}^{a}\cofr^{I}\,.
\end{equation}
Observe that, as the name suggests, these objects only cover the transversal part of the cotangent space (orthogonal to the 1-forms $\dfk$ and $\dfl$). Using this, an arbitrary element of the coframe can be expressed
\begin{equation}
\cofr^a =l^{a}\dfk+k^a \dfl+\spatcofr^{a}\,,
\end{equation}
The anholonomy two form associated to this coframe and the Levi-Civita connection 1-form of the metric are given 
  by,\footnote{Note that the last term, $-\frac{1}{2}(\tilde{\Omega}_{cab} +\tilde{\Omega}_{bca}-\tilde{\Omega}_{abc})
$, is indeed the Levi-Civita connection 1-form $\mathring{\tilde{\omega}}_{cab}$ associated to the metric $\tilde{g}_{ij}$ of the transversal sections.}
\begin{align}
\dex\cofr^a & =\left(-\tfrac{1}{2}k^a \partial_{b}H+\tilde{\Omega}_{b}{}^{a}\right)\dfk\wedge\spatcofr^{b}+\tfrac{1}{2}\tilde{\Omega}_{bc}{}^{a}\spatcofr^b\wedge\spatcofr^c \,,\\
\mathring{\dfom}_{ab}& =\left(\tilde{\partial}_{[a}Hk_{b]}-\tilde{\Omega}_{[ab]}\right)\dfk-2\tilde{\Omega}_{(cd)}\delta_{[a}^{d}k_{b]}\spatcofr^{c}-\tfrac{1}{2}\left(\tilde{\Omega}_{cab}+\tilde{\Omega}_{bca}-\tilde{\Omega}_{abc}\right)\spatcofr^{c}\,\label{eq: LC general}
\end{align}
where we have defined the transversal objects
\begin{align}
\tilde{\partial}_{a} & \coloneqq\delta_{a}^{I}\partial_{I}=(\delta_{a}^{b}-l^b k_a -k^b l_{a})\partial_{b}\,,\label{eq: transv partial}\\
\tilde{\Omega}_{a}{}^{b} & \coloneqq\delta_{a}^{J}e^{i}{}_{J}\delta_{I}^{b}\Omega_{ui}{}^{I}\,, \label{eq: ntriv anhol 1}\\
\tilde{\Omega}_{ab}{}^{c}&\coloneqq\delta_{a}^{J}\delta_{b}^{K}\delta_{I}^{c}\Omega_{JK}{}^{I} \label{eq: ntriv anhol 2}\,.
\end{align}

\subsection{Connection}

Based on the connections used in \cite{Obukhov2017, Obukhov2006}, we start by analyzing the following one
\begin{equation}
\dfom_{ab}=\mathring{\dfom}_{ab}+(\mathcal{C}_{ab}k_c +\mathcal{P}_{cab})\cofr^c +k_a k_b \dfA+g_{ab}\dfB\,, \label{eq: gen conn}
\end{equation}
where $\dfA=A_a\cofr^a$ and $\dfB=B_a\cofr^a$ are in principle general 1-forms, and $\mathcal{C}_{ab}$ and $\mathcal{P}_{cab}$ are arbitrary tensors satisfying
\begin{equation}
\mathcal{C}_{ab}=\mathcal{C}_{[ab]}\,,\qquad\mathcal{P}_{cab}=\mathcal{P}_{c[ab]}\,,\qquad k^{c}\mathcal{P}_{cab}=0=l^{c}\mathcal{P}_{cab}\,.
\end{equation}

Obviously $\dfom_{ab}$, as a whole, is completely independent of the metric. However we prefer to work with this decomposition into Levi-Civita plus distorsion because the degrees of freedom of the connection are stored within the tensorial objects $\mathcal{C}_{ab}$, $\mathcal{P}_{cab}$, $A_{a}$ and $ B_{a}$. Consequently, expressions like $\nabla_{\mu}\mathcal{C}_{ab}$ make sense. But the objects $\mathcal{C}_{ab}$ and $\mathcal{P}_{cab}$ are somehow metric-dependent because of the extraction of the Levi-Civita part. The actual explicitly non-metric expression of the connection is
\begin{equation}
\dfom_{ab}=(\mathsf{C}_{ab}k_c +\mathsf{P}_{cab})\cofr^c +k_a k_b \dfA+g_{ab}\dfB\,,
\end{equation}
where
\begin{align}
\mathsf{C}_{ab} & \coloneqq\mathcal{C}_{ab}+\tilde{\partial}_{[a}Hk_{b]}-\tilde{\Omega}_{[ab]}\,,\\
\mathsf{P}_{cab} & \coloneqq \mathcal{P}_{cab}-2\tilde{\Omega}_{(cd)} \delta_{[a}^{d}k_{b]}-\tfrac{1}{2}\left(\tilde{\Omega}_{cab} +\tilde{\Omega}_{bca}-\tilde{\Omega}_{abc}\right)\,.
\end{align}
Note that these two objects do not transform as tensors.\footnote{Working in terms of them is more cumbersome since, e.g. $\nabla_{\mu}\mathsf{C}_{ab}$ does not make sense.} Actually, the most important reason to avoid them is that, at the end of the day, the equations of motion of a covariant metric-affine Lagrangian can be written in terms of the curvature, the torsion and the non-metricity that can only depend on the combinations $\mathcal{C}_{ab}$ and $\mathcal{P}_{cab}$, due to their tensorial nature.

For future purposes we introduce the following decompositions
\begin{align}
A_{a}            &=\sA_{a}+\kA l_{a}+Ak_a \,,\\
B_{a}            &=\sB_{a}+\kB l_{a}+Bk_a \,,\\
\mathcal{C}_{ab} & = \sC_{ab}+2\kC_{[a}l_{b]}+2C_{[a}k_{b]}+2Ck_{[a}l_{b]}\,,\\
\mathcal{P}_{cab}& =\sP_{cab}+2\kP_{c[a}l_{b]}+2P_{c[a}k_{b]}+ 2P_{c}k_{[a}l_{b]}\,,
\end{align}
where the tensors $\sA_a$, $\sB_a$, $\kC_a$, $C_a$, $\sC_{ab}$, $P_c$, $P_{ca}$, $\kP_{ca}$ and $\sP_{cab}$ are totally transversal. The curvature, torsion and non-metricity as well as the irreducible decomposition of the last two are collected in Appendix \ref{app: RTQ general conn}.

\subsection{Summary and Lichnerowicz criteria}

Putting together all of the structures described in this section we have the following geometry
\begin{align}
\teng & =\cofr^{0}\otimes\cofr^{1}+\cofr^{1}\otimes \cofr^0-\delta_{IJ}\cofr^{I}\otimes\cofr^{J}\,\nonumber\\
\cofr^a &= \big\{ \cofr^0=\dfk=\dex u, \quad \cofr^1 = \dfl=\tfrac{1}{2}H(u,\,z)\dex u+\dex v, \quad \cofr^I = e_i{}^I(u,\,z)\dex z^i \big\} \,\nonumber\\
\dfom_{ab}&=\mathring{\dfom}_{ab}+(\mathcal{C}_{ab}k_c +\mathcal{P}_{cab})\cofr^c +k_a k_b \dfA+g_{ab}\dfB\,. \label{eq: typeI}
\end{align}

\begin{theorem} \label{thm: LichC typeI}
For a geometry of the type \eqref{eq: typeI}, the condition $\mathrm{1LCT}$ is equivalent to
\begin{equation}
0=P_c=\kP_{ca}=P_{[cd]}=\sP_{[cd]}{}^a= \kB = \sB_a \,;
\end{equation}
the condition $\mathrm{2LCT}$ is equivalent to
\begin{align}
0 & =\kC_a=\kP_{ca}= \kB\,,\nonumber\\
0 & = \sB_c - P_c\,, \nonumber \\
0 & = \kA + C - B \,;
\end{align}
the condition $\mathrm{1LCR}$ is equivalent to
\begin{align}
0&=\dfk\wedge\dex\dfB\,, \nonumber\\
0&=e^i{}_a e^j{}_b \partial_{[i}\sA_{j]}+2\sA_{[a}P_{b]}\,, \nonumber\\
0&=\partial_v \sA_a-\tilde{\partial}_a \kA+2 \kA P_a\,, \nonumber\\
0&=\kP_{ab} \kA\,, \nonumber\\
0&=\sA_{[a}\bar{P}_{b]c}\,,\nonumber\\
0&=\partial_{v}\mathcal{P}_{cab} \ (=k^d \mathring{\nabla}_{d} \mathcal{P}_{cab})\,, \nonumber\\
0&=\spatcofr^d\wedge\spatcofr^c \left(\mathring{\nabla}_{[d}\mathcal{P}_{c]ab}+\mathcal{P}_{[d|eb}\mathcal{P}_{|c]a}{}^{e}\right) \,;
\end{align}
and, finally, the condition $\mathrm{2LCR}$ is equivalent to
\begin{align}
0 & =\partial_v\mathcal{C}_{ab}\ (=k^d\mathring{\nabla}_d\mathcal{C}_{ab})\,,\nonumber\\
0 & =\partial_v\mathcal{P}_{cab}\ (=k^{d}\mathring{\nabla}_{d}\mathcal{P}_{cab})\,,\nonumber\\
0 & = 2 \partial_{[v}A_{u]} + 2C\kA = \partial_v (A-\tfrac{1}{2}H \kA) -\partial_u\kA + 2C\kA\,, \nonumber\\
0 & = \partial_v \sA_a-\tilde{\partial}_a \kA+2 \kA P_a\,,  \nonumber\\
0 & =\kC_{a}\kA\,,\nonumber\\
0 & =\bar{P}_{ab}\kA\,,\nonumber\\
0 & =\partial_v \sB_i - \partial_i \kB \,,\nonumber\\
0 & =2\partial_{[v} B_{u]} =\partial_v (B-\tfrac{1}{2}H\kB) - \partial_u \kB .
\end{align}
\end{theorem}

\begin{proof}
It follows straightforwardly from the application of Proposition \ref{prop: Lich crit and quad general} to our particular torsion and curvature 2-forms (see Appendix \ref{app: RTQ general conn}).
\end{proof}

 Moreover, combining the four conditions we easily arrive at
\begin{corollary} \label{cor: LichC typeI together}
For a geometry of the type \eqref{eq: typeI}, the generalized Lichnerowicz criteria for torsion and curvature (Definition \textup{\ref{def: GenLC}}) is verified if and only if
\begin{align}
0&=P_c=\kP_{ca}=P_{[cd]}=\sP_{[cd]}{}^a= \kB = \sB_a =\kC_a \nonumber\\
0& = \kA + C - B \,,\nonumber\\
0& =\partial_v B =\partial_v\mathcal{C}_{ab}=\partial_v\mathcal{P}_{cab} \,,\nonumber\\
0&=\partial_v \sA_a-\tilde{\partial}_a \kA\,, \nonumber\\
0&= \partial_v (A-\tfrac{1}{2}H \kA) -\partial_u\kA + 2C\kA\,,\nonumber\\
0&=\partial_{[i}\sA_{j]}\,, \nonumber\\
0&=\spatcofr^d\wedge\spatcofr^c\left(\mathring{\nabla}_{[d}\sP_{c]}{}^{ab} +2\mathring{\nabla}_{[d}P_{c]}{}^{[a}k^{b]}+ 2\sP_{[d}{}^{e[a}k^{b]}P_{c]e}+\sP_{[d}{}^{ea}\sP_{c]e}{}^{b}\right) \,.
\end{align}
\end{corollary}

As can be proved, the connections chosen in \cite{Obukhov2017, Obukhov2006} are LCR but not LCT, since the condition $0= \kA + C - B$ is not satisfied. For future convenience and due to we are interested in doing part of our calculations generalizing those papers, we introduce the abbreviation LCT* for those connections satisfying all of the conditions except that one. In addition, to make the violation of the LCT as explicit as possible we introduce the scalar function
\begin{equation}
\mathcal{Y}\coloneqq \kA+C-B\,. \label{eq: def Y LCT*}
\end{equation}
At this point, it is not difficult to see that under LCT* the violation of the quadratic condition for the torsion is indeed proportional to the square of it,
\begin{equation}
\dfT{}^a\wedge\star\dfT^b=-\mathcal{Y}^2 k^a k^b\ \volf\,,
\end{equation}
where $\volf\coloneqq\star 1= \sqrt{|g|}\dex u\wedge\dex v\wedge\dex z^2 ...\wedge z^{\dimM-1}$ is the canonical volume form associated to the metric $\teng$.

\section{Subcase: pp-waves and other simplifications}\label{sec: subcase}

If we do not impose additional conditions, we continue having a complicated Ansatz to search for solutions in many metric-affine theories. We are going to restrict further the theory by imposing two simplifications, one in the coframe (equivalently in the holonomic metric $g_{\mu\nu}$) and the other in the connection.

\subsection{Metric and coframe}

First we take the pp-wave case (see Definition \ref{def: ppwave}), i.e. we will assume the transversal space to be only $u$-dependent $\tilde{g}_{ij}(u,\,z)=\tilde{g}_{ij}(u)$, consequently there should be a redefinition of the transversal coordinates $z^i$ such that $\tilde{g}_{ij}$ becomes diagonal. Because of this, consider the case
\begin{equation}
\mathrm{d} s^2 =2\dex u\dex v+H(u,\,z)\dex u^2-\delta_{ij}\dex z^i\dex z^j\,.
\end{equation}
In the metric-affine context where we work in terms of the coframe, this is equivalent to
\begin{equation}
e_i{}^I =\delta_i^I\quad\Leftrightarrow\quad\cofr^I=\delta_i^I\dex z^i\,.
\end{equation}
Now, the anholonomy gets simplified and only $\Omega_{ui}{}^{1}=-\partial_i H$ survives, i.e.
\begin{equation}
\dex\cofr^a =-\tfrac{1}{2} k^a\tilde{\partial}_b H\dfk\wedge\cofr^b\,,
\end{equation}
so
\begin{equation}
\tilde{\Omega}_a{}^b=\tilde{\Omega}_{ab}{}^c=0\,. \label{eq: typeII zero anhol}
\end{equation}

\begin{prop}
In the gauge coframe, the Levi-Civita connection 1-form of the pp-wave metric is:
\begin{equation}
\mathring{\dfom}_{ab}=\tilde{\partial}_{[a}Hk_{b]}\dfk=\partial_{[a}Hk_{b]}\dfk\qquad\Leftrightarrow\qquad\mathring{\omega}_{\mu ab}=k_{\mu}\partial_{[a}Hk_{b]}\,.
\end{equation}
\end{prop}

\begin{proof}
Obtaining $\mathring{\dfom}_{ab}=\tilde{\partial}_{[a}Hk_{b]}\dfk$ is immediate starting from \eqref{eq: LC general}. To see the rest we need to expand the transversal derivative,
\[
\tilde{\partial}_{[a}Hk_{b]}=(\delta_{[a|}^{c}-l^{c}k_{[a|}-k^{c}l_{[a|})\partial_c Hk_{|b]}=\partial_{[a}Hk_{b]}\,.
\]
\end{proof}

When working with indices in the gauge base (because the connection is non-covariant object under frame transformations), we have
\begin{align}
k^{c}\mathring{\omega}_{cab} & =0\,,\\
l^{c}\mathring{\omega}_{cab} & =\tilde{\partial}_{[a}Hk_{b]}\,,\\
\mathring{\omega}_{ca}{}^{c}\equiv g^{cb}\mathring{\omega}_{cab} & =0\,.
\end{align}
Consequently for any totally transversal tensor   $\tilde{S}_{ab...}{}^{c...}$,
\begin{equation}
\tilde{\partial}^a\tilde{S}_{ab...}{}^{c...}=\partial^a\tilde{S}_{ab...}{}^{c...}=\mathring{\nabla}^{a}\tilde{S}_{ab...}{}^{c...}=\mathring{\nabla}_{a}\tilde{S}^a{}_{b...}{}^{c...}=\partial_{a}\tilde{S}^a{}_{b...}{}^{c...}=\tilde{\partial}_a\tilde{S}^a{}_{b...}{}^{c...}\,.
\end{equation}
Let us insist on that these equations are only true in the gauge basis because $\partial_c g_{ab}=0$ in that particular frame and, as a consequence of our basis choice, $k_a $, $l_a$, $k^a $ and $l^a$ are also constant. If we change the frame, the new anholonomy coefficients would enter the game.

Another consequence is that the covariant derivative of $l^a$ can be written in the following covariant way (valid in any frame, not only in the gauge basis),
\begin{equation}
\mathring{\nabla}_c l^a = \tfrac{1}{2} k_c (k^a l^b- g^{ab})\partial_bH\,
\end{equation}
which implies 
\begin{equation}
\mathring{\nabla}_c l^c = 0\,.
\end{equation}

Finally we present the Levi-Civita curvature and its irreducible parts in the gauge basis which gives very simple and practical expressions
\begin{align}
\mathring{\dfR}_{ab} & =\tilde{\partial}_{c}\tilde{\partial}_{[a}Hk_{b]}\spatcofr^{c}\wedge\dfk\,,\\
\mathring{\dfW}^{(1)}{}_{ab} & =\left(\tilde{\partial}_{c}\tilde{\partial}_{[a}Hk_{b]}-\tfrac{1}{\dimM-2}\tilde{\partial}^{e}\tilde{\partial}_{e}Hk_{[b}g_{a]c}\right)\spatcofr^{c}\wedge\dfk\,,\\
\mathring{\dfW}^{(4)}{}_{ab} & =\tfrac{1}{\dimM-2}\tilde{\partial}^{e}\tilde{\partial}_{e}Hk_{[b}g_{a]c}\spatcofr^{c}\wedge\dfk\,,\\
\mathring{\dfW}^{(6)}{}_{ab} & =0\,.
\end{align}
The non-trivial ones can be covariantized as follows
\begin{align}
\mathring{\dfR}_{ab} & =\mathring{\nabla}_{c}\partial_{[a}Hk_{b]}\cofr^c \wedge\dfk\,,\\
\mathring{\dfW}^{(1)}{}_{ab} & =\left(\mathring{\nabla}_{c}\partial_{[a}Hk_{b]}-\tfrac{1}{\dimM-2}\mathring{\nabla}^{2}Hk_{[b}g_{a]c}\right)\cofr^c \wedge\dfk\,,\\
\mathring{\dfW}^{(4)}{}_{ab} & =\tfrac{1}{\dimM-2}\mathring{\nabla}^{2}Hk_{[b}g_{a]c}\cofr^c \wedge\dfk\,,
\end{align}
where now the latin indices refer to any basis. These are totally $\mathrm{GL}(\dimM,\,\mathbb{R})$-covariant equations. Note that 
\begin{equation}
k^a\mathring{\dfR}_{ab}=k^a\mathring{\dfW}^{(1)}{}_{ab}=k^a\mathring{\dfW}^{(4)}{}_{ab}=0\,.
\end{equation}

\subsection{Connection}

In addition to the pp-wave condition we also impose on the connection the following metric-independent restriction in the gauge basis,
\begin{equation}
\mathsf{P}_{cab}=0\,.
\end{equation}
Note that we specify the basis because $\mathsf{P}_{cab}$ is not a tensor and this condition only holds in very particular frames. This, together with \eqref{eq: typeII zero anhol}, imply that $\mathcal{P}_{cab}=0$ and, therefore,
\begin{equation}
\dfom_{ab}=\mathring{\dfom}_{ab}+\mathcal{C}_{ab}\dfk+k_a k_b \dfA+g_{ab}\dfB\,, \label{eq: conn 2}
\end{equation}
whose torsion and curvature get simplified
\begin{align}
\dfT^a & =\mathcal{C}_c{}^a\dfk\wedge\cofr^c+k^a\dfA\wedge\dfk+\dfB\wedge\cofr^a\,,\nonumber\\
 & =\left[-Ck^a -\kC^{a}-\kA k^a +B k^a -\kB l^{a}\right]\dfk\wedge\dfl\nonumber \\
 & \quad+\left[\kC_c l^{a}+C_{c}k^a +\sC_{c}{}^{a}-k^a \sA_{c}+B\delta_{c}^{a}-\sB_{c}l^{a}\right]\dfk\wedge\spatcofr^c\nonumber \label{eq: T for type II}\\
 & \quad+\left[\kB \delta_{c}^{a}-\tilde{ B}_{c}k^a \right]\dfl\wedge\spatcofr^c +\sB_c\spatcofr^c\wedge\spatcofr^a\,.\\
\dfR_{a}{}^{b} & =\mathring{\dfR}_{a}{}^{b}+\Dex\mathcal{C}_a{}^b\wedge\dfk+k_a k^b\dex\dfA+\delta_{a}^{b}\dex\dfB \nonumber\\
 & =\mathring{\dfR}_a{}^b+\mathring{\Dex}\mathcal{C}_a{}^b\wedge\dfk+k_a k^b \dex\dfA+\delta_a^b\dex\dfB- k_c (k_a \mathcal{C}^{bc}+k^b \mathcal{C}_a{}^c)\dfk\wedge\dfA \label{eq: R for type II} \,,
\end{align}
while no changes in the non-metricity have been made with respect to that of the connection \eqref{eq: gen conn} (see Appendix \ref{app: RTQ general conn}). More details on them and their irreducible decomposition, as well as some other expressions, are presented for completeness in Appendix \ref{app: RTQ conn 2}.

\subsection{Summary and Lichnerowicz criteria}

Again we summarize the geometry we have considered in this section,
\begin{align}
\teng & =\cofr^{0}\otimes\cofr^{1}+\cofr^{1}\otimes \cofr^0-\delta_{IJ}\cofr^{I}\otimes\cofr^{J}\,\nonumber\\
\cofr^a &= \big\{ \cofr^0=\dfk=\dex u, \quad \cofr^1 = \dfl=\tfrac{1}{2}H(u,\,z)\dex u+\dex v, \quad \cofr^I = \delta_i^I\dex z^i \big\} \,\nonumber\\
\dfom_{ab}&=\mathring{\dfom}_{ab}+\mathcal{C}_{ab}k_c \cofr^c +k_a k_b \dfA+g_{ab}\dfB\,. \label{eq: typeII}
\end{align}

In this case the generalized Lichnerowicz criteria give essentially the same as in Theorem \ref{thm: LichC typeI} and Corollary \ref{cor: LichC typeI together} but setting $P_c=\kP_{ab}=P_{ab}=\sP_{abc}=\mathcal{P}_{abc}=0$.

\section{Further restrictions on the connection}\label{sec: subsubcase}

Finally, we consider an additional restriction of the subcase treated in the last section. The metric and the coframe continue being the same, but we consider the connection to be subjected to the following constraints 
\begin{align}\label{eq: conditions final conn}
C & =C(u) \,, &\dfB & = B(u)\dfk \,,\nonumber\\
\kC_a & =0 \,,             &0&=\partial_v \sA_a-\tilde{\partial}_a \kA\,, \nonumber\\
C_a & =C_a(u,\,z) \,, &0&=\partial_v (A-\tfrac{1}{2}H \kA) -\partial_u\kA + 2C\kA \,, \nonumber\\
\sC_{ab} & =\sC_{ab}(u)\,, &0&=\partial_{[i}\sA_{j]}\,.
\end{align}

A few remarks:
\begin{itemize}
\item We have as an immediate corollary $\dex\dfB=0$.
\item The Levi-Civita part remains the same as in \eqref{eq: typeII}. So the Riemannian
curvature is purely Weyl ($\mathring{\dfW}^{(1)}{}_{ab}$) and Ricci
($\mathring{\dfW}^{(4)}{}_{ab}$), while the curvature scalar vanishes.
\item This configuration together with the conditions
\begin{equation}
C_a=\sC_{ab}=C=B=0
\end{equation}
reproduces the Ansatz for the connection in \cite{Obukhov2006}. If, instead, we impose 
\begin{equation}
\sA_a=A=\kA=\sC_{ab}=C=B=0\,,
\end{equation}
we obtain the one in \cite{Obukhov2017}.
\end{itemize}

\begin{theorem}
The geometry \eqref{eq: typeII} together with \eqref{eq: conditions final conn} satisfies both \textup{LCR} and \textup{LCT*}. If, additionally, $\mathcal{Y}=0$, then the full generalized Lichnerowicz criterion is fulfilled.
\end{theorem}

Let us now focus on the basic tensors associated to the connection and the properties they acquire under \eqref{eq: conditions final conn}. From now on we will use \eqref{eq: def Y LCT*} to eliminate $\kA$ from all the equations.

\subsection{Torsion and its properties}

Now we present the torsion and its irreducible components,
\begin{align}
\dfT^a & =-\mathcal{Y}k^a \dfk\wedge\dfl+\big[\sC_{c}{}^{a}+(C_{c}- \sA_{c})k^a + B\delta_{c}^{a}\big]\dfk\wedge\spatcofr^c \\
\dfT^{(\mathrm{tr})}{}^{a} & =\tfrac{1}{\dimM-1}\left[(\dimM-2) B-\mathcal{Y}\right]\dfk\wedge\cofr^a \,,\\
\dfT^{(\mathrm{a})}{}^{a} & =\tfrac{1}{3}\sC_{bc}\left(k^a \spatcofr^b\wedge\spatcofr^c+2g^{ac}\dfk\wedge\spatcofr^{b}\right)\\
\dfT^{(\mathrm{tn})}{}^{a} & =\left[\tfrac{1}{3}\sC_{b}{}^a +(C_b- \sA_{b})k^a +\tfrac{1}{\dimM-1}(\mathcal{Y}+ B)\delta_b^a\right]\dfk\wedge\spatcofr^b\\
 & \quad-\tfrac{\dimM-2}{\dimM-1}(\mathcal{Y}+ B)k^a \dfk\wedge\dfl-\tfrac{1}{3}\sC_{bc}k^a \spatcofr^b\wedge\spatcofr^c
\end{align}
It is worth remarking that the totally antisymmetric component is directly connected with the antisymmetric transversal tensor $\sC_{ab}$, and that the trace of the torsion,
\begin{equation}
-\dint{\vfre_a}\dint{\vfre_b}\dfT^{(\mathrm{tr})}{}^b=T_{ab}{}^b=\left[(D-2)B-\mathcal{Y}\right]k_a \,,
\end{equation}
 is proportional to $B$ for geometries LCT (i.e. with $\mathcal{Y}=0$).

\begin{itemize}
\item The operator $\dfk\wedge$ gives zero also when acting on $\dfT^{(\mathrm{tr})}{}^{a}$ since it is proportional to $\dfk$, but for the rest of the irreducible components we have
\begin{equation}
\dfk\wedge\dfT^{(\mathrm{a})}{}^{a}=-\dfk\wedge\dfT^{(\mathrm{tn})}{}^{a}=\tfrac{1}{3}\sC_{bc}k^a \dfk\wedge\spatcofr^b\wedge\spatcofr^c\,.
\end{equation}

\item However the operator $\dfk\wedge\star$ gives zero for $\dfT^{(\mathrm{a})}{}^{a}$ because it only has components in the directions $\dfk$ and $\spatcofr^{c}$. For the rest,
\begin{align}
\dfk\wedge\star\dfT^{(\mathrm{tr})}{}^{a} & =\tfrac{1}{\dimM-1}\left[(\dimM-2) B-\mathcal{Y}\right]k^a \star\dfk\,,\\
\dfk\wedge\star\dfT^{(\mathrm{tn})}{}^{a} & =-\tfrac{\dimM-2}{\dimM-1}(\mathcal{Y}+ B)k^a \star\dfk\,.
\end{align}
\item The contractions of the torsion 2-form and its irreducible components with $k^a $ all vanish,
\begin{equation}
k_a \dfT^{(\mathrm{a})}{}^a=k_a \dfT^{(\mathrm{tr})}{}^a=k_a \dfT^{(\mathrm{tn})}{}^{a}=0\qquad\Rightarrow\qquad k_a \dfT{}^a =0\,.
\end{equation}
This is not true for contractions in the first two indices of the torsion tensor $T_{bc}{}^a$ and their irreducible components, since there are coefficients in the direction of $\dfl$, which give 1 instead of 0 when contracting with $k^a$.

\end{itemize}

\subsection{Non-metricity and its properties}

Now we have
\begin{equation}
\dfQ_{ab}=2 k_a k_b \dfA+2g_{ab} B\dfk\,.
\end{equation}
Therefore the traces are
\begin{align}
\dfQ_{c}{}^{c} =Q_{ac}{}^c \cofr^a & = 2\dimM B\dfk\,,\\
\dint{\vfre^{c}}\dfQ_{ca}=Q_{ca}{}^c &= 2(\mathcal{Y}+2 B-C)k_a \,
\end{align}
The irreducible components of $\dfQ_{ab}$ do not experience any changes with respect to those in Appendix \ref{app: RTQ general conn}, apart from the substitution $\dfB=B\dfk$. 

Let us show some properties of this non-metricity:

\begin{itemize}
\item Contractions of the non-metricity with $k^a$
\begin{align}
k^a \dfQ{}_{ab}=k^a \dfQ^{(\mathrm{tr}1)}{}_{ab} & =2k_b  B\dfk\,,\\
k^a \dfQ^{(\mathrm{tr}2)}{}_{ab} & =\tfrac{2(\dimM-2)}{(\dimM-1)(\dimM+2)}(\mathcal{Y}+ B-C)k_b \dfk\,,\\
k^a \dfQ^{(\mathrm{s})}{}_{ab} & =\tfrac{2(\dimM-2)}{3(\dimM+2)}(\mathcal{Y}+ B-C)k_b \dfk\,,\\
k^a \dfQ^{(\mathrm{tn})}{}_{ab} & =-\tfrac{2(\dimM-2)}{3(\dimM-1)}(\mathcal{Y}+ B-C)k_b \dfk\,.
\end{align}
An immediate consequence is
\begin{equation}
k^a k^b \dfQ^{(N)}{}_{ab}=0\qquad\qquad\forall N\,.
\end{equation}
\item Contracted wedge with the coframe\footnote{Here we are  extracting the antisymmetric part of the corresponding tensors in the first two indices, $Q^{(N)}{}_{[ca]b}$.}
\begin{align}
\dfQ^{(\mathrm{tr}1)}{}_{ab}\wedge\cofr^a  & =-2 B\cofr_b \wedge\dfk\,,\label{eq: Q wedge cofr 1}\\
\dfQ^{(\mathrm{tr}2)}{}_{ab}\wedge\cofr^a  & =\tfrac{2}{\dimM-1}(\mathcal{Y}+ B-C)\cofr_b \wedge\dfk\,,\\
\dfQ^{(\mathrm{tn})}{}_{ab}\wedge\cofr^a  & =2k_b \dfA\wedge\dfk-\tfrac{2}{\dimM-1}(\mathcal{Y}+ B-C)\cofr_b \wedge\dfk\,,\\
\dfQ^{(\mathrm{s})}{}_{ab}\wedge\cofr^a  & =0\,.\label{eq: Q wedge cofr 4}
\end{align}
\item Finally we apply the operators $\dfk\wedge $ and $\dfk\wedge\star $ on the non-metricity form,
\begin{align}
\dfk\wedge\dfQ_{ab} & =2k_a k_b (\mathcal{Y}+ B-C)\dfk\wedge\dfl+2k_a k_b  A_{c}\dfk\wedge\spatcofr^{c}\,,\\
\dfk\wedge\star\dfQ_{ab} & =2k_a k_b (\mathcal{Y}+ B-C)\volf\,.
\end{align}
\end{itemize}

\subsection{Curvature and its properties}

Only the following parts of the curvature survive to the conditions \eqref{eq: conditions final conn}
\begin{align}
\dfZ^{(1)}{}_{ab} & =k_a k_b(2C \sA_c-l^d \calF_{dc})\spatcofr^c\wedge\dfk\,,\\
\dfW^{(1)}{}_{ab} & =\left[2\VdHC_{(cd)}-\tfrac{2}{\dimM-2}g_{cd}\VdHC\right]\delta_{[a}^{d}k_{b]}\spatcofr^c\wedge\dfk\,,\\
\dfW^{(2)}{}_{ab} & =2\VdHC_{[cd]}\delta_{[a}^{d}k_{b]}\spatcofr^c \wedge\dfk\,,\\
\dfW^{(4)}{}_{ab} & =\tfrac{2}{\dimM-2}\VdHC g_{c[a}k_{b]}\spatcofr^c\wedge\dfk\,,
\end{align}
where we have introduced the following transversal tensors
\begin{align}
\VdHC_{a} & \coloneqq\tfrac{1}{2}\tilde{\partial}_a H+C_a\,,\label{eq: objects R type III 1}\\
\VdHC_{ab} & \coloneqq\tilde{\partial}_a\VdHC_b=\mathring{\nabla}_b\VdHC_a-k_b l^c \mathring{\nabla}_c\VdHC_a\,,\\
\VdHC & \coloneqq\VdHC_c{}^c =\tilde{\partial}_c \VdHC^c =\mathring{\nabla}_c\VdHC^c\,,
\end{align}
which are going to dominate the antisymmetric part of the curvature $\dfW{}_{ab}$ and its parts, and the antisymmetric tensor
\begin{equation}
\calF_{dc}\coloneqq2\mathring{\nabla}_{[d} A_{c]}=2e^{\mu}{}_{d}e^{\nu}{}_{c}\partial_{[\mu} A_{\nu]}\,.\label{eq: objects R type III 2}
\end{equation}
These new objects fulfill the relations
\begin{equation}
\tilde{\partial}^{c}\VdHC_{cd}=\tilde{\partial}^{2}\VdHC_{d}\,,\qquad\tilde{\partial}^{c}\VdHC_{dc}=\tilde{\partial}_{d}\VdHC^{c}{}_{c}=\tilde{\partial}_{d}\VdHC\,,
\end{equation}
\begin{equation}
l^d \calF_{dc} \spatcofr^c = (\partial_u\sA_c - \tilde{\partial}_c A) \spatcofr^c \,.
\end{equation}

As a consequence, the final expression for the total curvature is
\begin{equation}
\dfR_{ab}=k_a k_b (2C \sA_c-l^d\calF_{dc})\spatcofr^c\wedge\dfk+2\VdHC_{c[a}k_{b]}\spatcofr^c\wedge\dfk\,. 
\end{equation}
Observe that from this expression we can immediately read the symmetric ($\dfZ{}_{ab}$) and antisymmetric ($\dfW{}_{ab}$) parts of the curvature. Notice also that the antisymmetric one is totally controlled by the tensor $\tilde{V}_{ca}$, while the symmetric part only depends on $A_a$ and its derivatives. Finally we provide some nice properties of this curvature:

\begin{itemize}
\item The Lichnerowicz conditions also hold independently for $\dfZ^{(1)}{}_{ab}$, $\dfW^{(1)}{}_{ab}$, $\dfW^{(2)}{}_{ab}$ and $\dfW^{(4)}{}_{ab}$ since all of them are linear combinations of $\spatcofr^c\wedge\dfk$.
\item In addition to the Lichnerowicz condition, we also have
\begin{equation}
k^a \dfZ_{ab}=k^a \dfW^{(1)}{}_{ab}=k^a \dfW^{(2)}{}_{ab}=k^a \dfW^{(4)}{}_{ab}=0\,.
\end{equation}
thanks to the fact that $k^aC_a=0=k_c\spatcofr^c$. Consequently
\begin{equation}
k^a \dfR_a{}^b{}=0 \,,\qquad k_b\dfR_a{}^b=0\,.
\end{equation}
This result together with the fact that $\dfR_{ab}$ goes in the direction of $\spatcofr^c\wedge\dfk$ tells that any contraction of the curvature tensor $R_{abcd}$ (and hence of any of its irreducible components) with the wave vector $k^a$ vanishes.

\item The traces of the curvature are
\begin{align}
\dint{\vfre^{b}}\dfR_{ba} & =\dint{\vfre^{b}}\dfW^{(4)}{}_{ba}=\VdHC k_a \dfk\,,\\
\dint{\vfre^{b}}\dfR_{ab} & =-\dint{\vfre^{b}}\dfW^{(4)}{}_{ba}=-\VdHC k_a \dfk\,,\\
\dfR_a{}^a & = 0\,,\\
\dint{\vfre^{a}}\dint{\vfre^{b}}\dfR_{ba} & =0\,.
\end{align}
or, equivalently in components,
\begin{align}
R_{acb}{}^c = - R_{ac}{}^{c}{}_b & = \VdHC k_a k_b  \,,\\
R_{abc}{}^c &= 0  \,,\\
R_{ab}{}^{ab} & = 0\,.
\end{align}
\end{itemize}

\subsection{Evaluated MAG Lagrangian}

Finally we would like to end this section with a useful result when working in Metric-Affine Gauge gravity. Consider the (even) MAG Lagrangian in arbitrary dimensions,
\begin{align}
\dfL_{(\dimM)}^{+} & =-\frac{1}{2\kappa}\bigg[2\lambda\volf-a_{0}\dfR^{ab}\wedge\star(\cofr_a\wedge\cofr_b)+\dfT^a\wedge\star\big(a_{1}\dfT^{(\mathrm{tn})}{}_{a}+a_{2}\dfT^{(\mathrm{tr})}{}_{a}+a_{3}\dfT^{(\mathrm{a})}{}_{a}\big)\nonumber \\
 & \qquad\qquad+\dfQ_{ab}\wedge\star\big(b_{1}\dfQ^{(\mathrm{s})}{}^{ab}+b_{2}\dfQ^{(\mathrm{tn})}{}^{ab}+b_{3}\dfQ^{(\mathrm{tr}2)}{}^{ab}+b_{4}\dfQ^{(\mathrm{tr}1)}{}^{ab}+b_{5}\dfQ^{(\mathrm{tr}1)}{}^{bc}{}_{c}\cofr^a \big)\nonumber \\
 & \qquad\qquad+2\big(c_{2}\dfQ^{(\mathrm{tn})}{}_{ab}+c_{3}\dfQ^{(\mathrm{tr}2)}{}_{ab}+c_{4}\dfQ^{(\mathrm{tr}1)}{}_{ab}\big)\wedge\cofr^a \wedge\star\dfT^b \bigg]\nonumber \\
 & \quad-\frac{1}{2\rho}\dfR_{a}{}^{b}\wedge\star\Bigg[\sum_{N=1}^{6}w_{N}\dfW^{(N)}{}^{a}{}_{b}+w_{7}\cofr^a \wedge\big(\dint{\vfre_{c}}\dfW^{(5)}{}^{c}{}_{b}\big)\nonumber \\
 & \qquad\qquad+\sum_{N=1}^{5}z_{N}\dfZ^{(N)}{}^{a}{}_{b}+z_{6}\cofr_c \wedge\big(\dint{\vfre^{a}}\dfZ^{(2)}{}^{c}{}_{b}\big)+\sum_{N=7}^{9}z_{N}\cofr^a \wedge\big(\dint{\vfre_{c}}\dfZ^{(N-4)}{}^{c}{}_{b}\big)\Bigg]\,,\label{eq: MAG even action}
\end{align}
and the 4-dimensional odd parity extension,
\begin{align}
\dfL_{(4)}^{-} & =-\frac{1}{2\kappa}\Big[a_{0}^{-}\dfR^{ab}\wedge\cofr_a\wedge\cofr_b+a_{1}^{-}\dfT^a\wedge\dfT^{(\mathrm{tn})}{}_{a}+2a_{2}^{-}\dfT^a\wedge\dfT^{(\mathrm{tr})}{}_{a}+b^{-}\dfQ_{b}{}^{c}\wedge\dfQ^{(\mathrm{tn})}{}_{ac}\wedge\cofr^a\wedge\cofr^b\nonumber \\
 & \quad\qquad\qquad+2c_{1}^{-}\dfQ^{(\mathrm{tr}1)}{}_{c}{}^{c}\wedge\cofr_a \wedge\dfT^{(\mathrm{a})}{}^{a}+2c_{2}^{-}(\dint{\vfre^{c}}\dfQ_{cb})\wedge\cofr^b \wedge\cofr_a\wedge\dfT^{(\mathrm{a})}{}^{a}\nonumber \\
 & \quad\qquad\qquad+2\dfQ_{ab}\wedge\cofr^a \wedge\left(c_{3}^{-}\dfT^{(\mathrm{a})}{}^{b}+c_{4}^{-}\dfT^{(\mathrm{tn})}{}^{b}\right)\Big]\nonumber\\
 & \quad-\frac{1}{2\rho}\dfR^{ab}\wedge\bigg[w_{1}^{-}\dfW^{(1)}{}_{ab}+2w_{2}^{-}\dfW^{(2)}{}_{ab}+2w_{3}^{-}\dfW^{(3)}{}_{ab}+w_{5}^{-}\dfW^{(5)}{}_{ab}\nonumber \\
 & \quad\qquad\qquad\qquad\qquad\qquad+z_{1}^{-}\dfZ^{(1)}{}_{ab}+2z_{2}^{-}\dfZ^{(2)}{}_{ab}+z_{3}^{-}\dfZ^{(3)}{}_{ab}+z_{4}^{-}\dfZ^{(4)}{}_{ab}\bigg]\,,\label{eq: MAG odd action}
\end{align}
where $\kappa$ and $\rho$ are the gravitational weak and strong coupling constants, and the rest are free dimensionless parameters.

\begin{theorem}
Let $\mathcal{G}=(g_{ab},\,\cofr^a ,\,\dfom_{a}{}^{b})$ be a geometry of the type treated this section, i.e. \eqref{eq: typeII} under the restrictions \eqref{eq: conditions final conn}. Then, in arbitrary dimensions, any even parity linear or quadratic invariant involving exclusively the curvature, the torsion and the non-metricity of the connection, and no derivatives of them, is identically zero. Furthermore, the 4-dimensional odd parity invariants that satisfy the previous requirements also vanish.
\end{theorem}

\begin{proof}
First we use that the irreducible components $\dfZ^{(I)}$ with $I=2,3,4,5$ and $\dfW^{(I)}$ with $I=3,5,6$ are zero. Then using the properties in the previous subsections, it is almost immediate to check that all of the terms that appear in \eqref{eq: MAG even action} vanish independently. Since they form a basis of all possible (linear and quadratic) invariants involving the curvature, the torsion and the non-metricity, then all possible invariants of this order are zero. Something similar happens with the basis of odd invariants in four dimensions built with the terms appearing in \eqref{eq: MAG odd action}.
\end{proof}

In particular, when looking for solutions of this type for the 4-dimensional MAG action $\dfL_{(4)}=\dfL_{(4)}^{+}+\dfL_{(4)}^{-}$, only the cosmological constant term contributes to the evaluated Lagrangian,
\begin{equation}
\left.\dfL_{(4)}\right|_{\mathcal{G}}= -\frac{\lambda}{\kappa}\volf\,.
\end{equation}
This result simplifies considerably the equation of motion of the coframe. To be precise, the term with the interior derivative of the Lagrangian reduces to\footnote{The term $\dint{\vfre_{a}}\dfL$ is in fact the one that comes from the variation of $\sqrt{|g|}$ with respect to the metric in the $(g_{\mu\nu},\,\Gamma_{\mu\nu}{}^{\rho})$ formulation.}
\begin{equation}
\dint{\vfre_{a}}\big(\left.\dfL_{(4)}\right|_{\mathcal{G}}\big)=-\frac{\lambda}{\kappa}\star\cofr_a \,.
\end{equation}

\section{Final comments}\label{sec: conclusion}

In this paper we revised several criteria that can be found in the literature to discern whether a Riemannian spacetime or a region of it belongs to a gravitational wave category, i.e. it contains gravitational radiation. We also recalled that, in the context of General Relativity, some of them are equivalent in vacuum for very simple kinds of metrics. Then we discussed some possibilities for them to be extended to a metric-affine geometry and focused on one of them, the Lichnerowicz criteria. The main motivation for this choice is that this criterion reflects some common features between electromagnetic radiation and gravitational waves. We therefore proposed a generalization of it and showed its implications for a particular geometry. For the metric (or, equivalently the pair formed by the anholonomic metric and the coframe) we considered a Brinkmann space, whereas the linear connection was chosen as a generalization of those studied in the works \cite{Obukhov2006, Obukhov2017}. We then collected the conditions this connection should satisfy in order to respect the proposed generalization of the Lichnerowicz criteria. Finally, we analyzed some particular cases providing several properties of the associated curvature, torsion and non-metricity.

At this point, one important remark is that we have concentrated here on generalizing the criteria used in Riemannian geometry, but there are other conditions to be taken into account, for instance, the symmetries of the metric (\emph{isometries}). In the Brinkmann case, the wave vector $\vpartial_v=k^\mu \vpartial_\mu$ is indeed a Killing vector, which can be seen in the fact that none of the metric components in the Brinkmann chart depends on the $v$ coordinate. Encouraged by this fact, one may also require the linear connection to have zero Lie derivative in the direction of $k^\mu$. Since this is true for the Levi-Civita part, it will be guaranteed whenever the distorsion tensor has zero Lie derivative. For instance, for our configuration \eqref{eq: typeII} expressed in the Brinkmann chart, this condition gives essentially
\begin{equation}
0= \partial_v\mathcal{C}_{\mu\nu}k^{\rho}+k_{\mu}k_{\nu}\partial_vA^\rho + g_{\mu\nu}\partial_v B^\rho\,.
\end{equation}
Contracting appropriately this equation one obtains that all of the tensors that the connection depends on must be $v$-independent. For $\mathcal{C}_{\mu\nu}$ this is true under the generalized Lichnerowicz criteria, but for $A^\mu$ and $B^\mu$ we get new conditions to be considered, which will obviously simplify further our geometries.

It is also worth mentioning the role of the metric in theories beyond General Relativity. The criteria explained in \cite{Zakharov1973} are defined in the context of the differential equations of motion of General Relativity. So in order for our metric Ansatz to be associated to gravitational radiation (in the sense of Lichnerowicz) it should be guaranteed that the equations of motion of the theory for the metric sector are of the same type. In the MAG case, this is true e.g. if the parameters of the action are such that the Riemannian (Levi-Civita) quadratic part in the curvature gives the Gauss-Bonnet invariant. In that case, so the four dimensional theory becomes simply General Relativity plus additional fields (torsion, non-metricity and their derivatives). The compatibility of the criteria with other theories that do not respect these requirements should be carefully studied.  In addition, the precise physical meaning of the generalized Lichnerowicz criterion (Definition \ref{def: GenLC}) in relation to the dynamical equations for the connection in each particular theory is another important question to address. These points and their implications in the MAG theory are aspects that we leave for future research.

The author is currently exploring the dynamics derived from the MAG action \eqref{eq: MAG even action} in arbitrary dimensions and the four dimensional case together with the odd parity terms \eqref{eq: MAG odd action}, searching for solutions of the type analyzed in these pages. 

%%%%%%%%%%%%%%%%%%%%%%%%%%%%%%%%%%%%%%%%%%%%%%%%%%%%%%%%
%%%%%%%%%%%%%%%%%%%%%%%%%%%%%%%%%%%%%%%%%%%%%%%%%%%%%%%%
%%%%%%%%%%%%%%%%%%%%%%%%%%%%%%%%%%%%%%%%%%%%%%%%%%%%%%%%
%\newpage
\vspace{.6cm}
\noindent
{\bf Acknowledgements}\\
The author would like to thank Bert Janssen, Tomi Koivisto, Adri\`a Delhom and Jos\'e Beltr\'an for their useful comments and specially Christian Pfeifer and Yuri Obukhov for helpful discussions and feedback. The author is supported by a PhD contract of the program FPU 2015 with reference FPU15/02864  of the Spanish Ministry of Economy and Competitiveness, which also funded this work through the project FIS2016-78198-P.

%\newpage
%%%%%%%%%%%%%%%%%%%%%%%%%%%%%%%%%%%%%%%%%%%%%%%%%%%%%%%
\appendix
%%%%%%%%%%%%%%%%%%%%%%%%

\section{Table of symbols} \label{app: symbols}

Here we collect our notation providing a table with the symbols we have used throughout the paper and the ranks of the differential forms.

\scriptsize
\renewcommand\arraystretch{1.2}

\noindent
\begin{tabular}{|>{\centering}p{0.15\columnwidth}|>{\centering}p{0.16\columnwidth}|>{\centering}p{0.5\columnwidth}|>{\centering}p{0.1\columnwidth}|}
\hline 
Differential form / tensor notation & Components & Meaning & Rank as diff. form\tabularnewline
\hline 
\multicolumn{4}{|c|}{Dimension, metric and basis}\tabularnewline
\hline 
 & $\dimM$ & Dimension of the manifold & \tabularnewline
$\teng$ & $g_{\mu\nu}$, $g_{ab}$ & Metric & \tabularnewline
$\volf$ &  & Canonical volume form associated to the metric (i.e. $\star1$) & $\dimM$\tabularnewline
$\vpartial_{\mu}$ &  & Coordinate frame & \tabularnewline
$\dex x^{\mu}$ &  & Coordinate coframe & 1\tabularnewline
$\vfre_{a}$ & $e^{\mu}{}_{a}$ & General (or gauge) frame & \tabularnewline
$\cofr^{a}$ & $e_{\mu}{}^{a}$ & General (or gauge) coframe & 1\tabularnewline
$\dex\cofr^{a}$ & $\Omega_{ab}{}^{c}$ & Anholonomy form/coefficients & 1\tabularnewline
\hline 
\multicolumn{4}{|c|}{Linear connection and associated objects}\tabularnewline
\hline 
$\dfom_{a}{}^{b}$ & $\Gamma_{\mu\nu}{}^{\rho}$, $\omega_{\mu a}{}^{b}$ & Linear connection & 1\tabularnewline
$\Dex\dfal$ &  & Exterior covariant derivative associated to the connection & $\text{rank}(\dfal)+1$\tabularnewline
 & $\nabla_{\mu}$ & Covariant derivative associated to the linear connection & \tabularnewline
$\dfQ_{ab}$ & $Q_{\mu\nu\rho}$, $Q_{\mu ab}$ & Nonmetricity associated to the linear connection & 1\tabularnewline
$\dfT{}^{a}$ & $T_{\mu\nu}{}^{\rho}$, $T_{\mu\nu}{}^{a}$,  & Torsion associated to the linear connection & 2\tabularnewline
$\dfR_{a}{}^{b}$ & $R_{\mu\nu\rho}{}^{\lambda}$, $R_{\mu\nu a}{}^{b}$ & Curvature associated to the linear connection & 2\tabularnewline
$\dfW_{a}{}^{b}$ & $W_{\mu\nu\rho}{}^{\lambda}$, $W_{\mu\nu a}{}^{b}$ & Antisymmetric part of the curvature (last two indices) & 2\tabularnewline
$\dfZ_{a}{}^{b}$ & $Z_{\mu\nu\rho}{}^{\lambda}$, $Z_{\mu\nu a}{}^{b}$ & Symmetric part of the curvature (last two indices) & 2\tabularnewline
$\dfQ^{(I)}{}_{ab}$, $\dfT^{(I)}{}^{a}$ &  & Irreducible components of the torsion and non-metricity & (see above) \tabularnewline
$\dfZ^{(I)}{}_{a}{}^{b}$, $\dfW^{(I)}{}_{a}{}^{b}$ &  & Irreducible components of the curvature & (see above) \tabularnewline
$\mathring{\dfom}_{a}{}^{b}$, $\mathring{\Dex}$,... &  $\mathring{\omega}_{\mu a}{}^{b}$,
$\mathring{\nabla}_{\mu}$,... & Levi-Civita connection and associated objects & (see above)\tabularnewline
\hline 
\multicolumn{4}{|c|}{Other operators on differential forms}\tabularnewline
\hline 
$\dex\dfal$ &  & Exterior derivative & $\text{rank}(\dfal)+1$\tabularnewline
$\dint{\vfre_a}\dfal$ &  & Interior product by the vector $\vfre_a$ & $\text{rank}(\dfal)-1$\tabularnewline
$\star\dfal$ &  & Hodge star & $\dimM-\text{rank}(\dfal)$\tabularnewline
\hline 
\multicolumn{4}{|c|}{Optical decomposition}\tabularnewline
\hline 
 & $h^{\mu}{}_{\nu}$ & Transversal proyector & 0\tabularnewline
 & $B_{\mu\nu}$ & Covariant derivative of the velocity & 0\tabularnewline
 & $\tilde{B}_{\mu\nu}$ & Transversal part of $B_{\mu\nu}$ & 0\tabularnewline
 & $\omega_{\mu\nu}$, $\sigma_{\mu\nu}$ & Twist and shear tensors & 0\tabularnewline
 & $\theta$, $\omega$, $|\sigma|$ & Optical scalars (expansion, twist and shear, respectively) & 0\tabularnewline
\hline 
\multicolumn{4}{|c|}{Objects used in the geometry we studied}\tabularnewline
\hline 
$\dfA$, $\dfB$, - , -  & $A_{a}$, $B_{a}$, $\mathcal{C}_{ab}$, $\mathcal{P}_{cab}$  & Tensorial objects in the connection \eqref{eq: typeI} & 1\tabularnewline
 & $A$, $\kA$, $\sA_{a}$ & Pieces in the decomposition of $A_{a}$ & 0\tabularnewline
 & $B$, $\kB$, $\sB_{a}$ & Pieces in the decomposition of $B_{a}$ & 0\tabularnewline
 & $C$, $C_{a}$, $\kC_{a}$, $\sC_{ab}$ & Pieces in the decomposition of $\mathcal{C}_{ab}$ & 0\tabularnewline
 & $P_{c}$, $P_{ca}$, $\kP_{ca}$, $\sP_{cab}$ & Pieces in the decomposition of $\mathcal{P}_{cab}$ & 0\tabularnewline
 & $H$, $W_i$ & Functions in the Brinkmann (or Kundt) metric & 0\tabularnewline
 & $\{u,\,v,\,z^{i}\}$ & Brinkmann (or Kundt) coordinates & \tabularnewline
$\dfk$ & $k_{\mu}$, $k_{a}$ & Wave form (or wave vector when the indices are raised) & 1\tabularnewline
$\dfl$ & $l_{\mu}$, $l_{a}$ & Lightlike form independent of $\dfk$ with $l^\mu k_\mu=1$ & 1\tabularnewline
  & $\mathcal{Y}$ & Scalar parameterizing the violation of the LCT [see \eqref{eq: def Y LCT*}] &0 \tabularnewline
$\dfF$ & $F_{\mu\nu}$ & Electromagnetic 2-form/ tensor & 2\tabularnewline
 & $\tilde{g}_{ij}$, $\tilde{g}_{IJ}$ & Transversal metric  & 0\tabularnewline
$\spatcofr^a$ &  & Transversal coframe & 1 \tabularnewline
 & $\tilde{\partial}_{a}$ & Transversal partial derivative  & \tabularnewline
 & $\tilde{\Omega}_{a}{}^{b}$, $\tilde{\Omega}_{ab}{}^{c}$ & Transversal anholonomy coefficients [see \eqref{eq: ntriv anhol 1} and \eqref{eq: ntriv anhol 2}] & 0 \tabularnewline
 & $\calF_{ab}$, $\VdHC$, $\VdHC_{a}$, $\VdHC_{ab}$ & See definitions \eqref{eq: objects R type III 1}-\eqref{eq: objects R type III 2} & 0\tabularnewline
\hline 
\end{tabular}

\renewcommand\arraystretch{1}
\normalsize

\begin{center}
Table A.1. Symbols and ranks of differential forms.
\end{center}

\section{Useful expressions} \label{app: useful expr}

\subsection{Optical decomposition}\label{app:opticaldecom}

For a lightlike congruence with velocity $k^\mu$ and for any lightlike vector $l^\mu$ such that $k^\mu l_\mu=\epsilon=\pm1$, the general covariant expressions for the twist tensor, the expansion and the shear tensor are
\begin{align}
\theta & =\tfrac{1}{\dimM-2}\left(\mathring{\nabla}_{\sigma}k^{\sigma}-\epsilon l_{\sigma}\dot{k}^{\sigma}\right)\,,\\
\omega_{\mu\nu} & =\mathring{\nabla}_{[\nu}k_{\mu]}-\epsilon(l^{\sigma}\mathring{\nabla}_{\sigma}k_{[\mu})k_{\nu]}+\epsilon l_{[\mu}\dot{k}_{\nu]}-\epsilon k_{[\mu}(l_{|\sigma|}\mathring{\nabla}_{\nu]}k^{\sigma})+k_{[\mu}l_{\nu]}l_{\sigma}\dot{k}^{\sigma}\,,\\
\sigma_{\mu\nu} & =\left[\mathring{\nabla}_{(\mu}k_{\nu)}-\tfrac{1}{\dimM-2}h_{\mu\nu}\mathring{\nabla}_{\sigma}k^{\sigma}\right]-\epsilon(l^{\sigma}\mathring{\nabla}_{\sigma}k_{(\mu})k_{\nu)}-\epsilon k_{(\mu}(l_{|\sigma|}\mathring{\nabla}_{\nu)}k^{\sigma})\nonumber \\
 & \qquad-\epsilon\left[l_{(\mu}\dot{k}_{\nu)}-\tfrac{1}{\dimM-2}h_{\mu\nu}l_{\sigma}\dot{k}^{\sigma}\right]+k_{(\mu}l_{\nu)}l_{\sigma}\dot{k}^{\sigma}+k_{\mu}k_{\nu}(l_{\lambda}l^{\sigma}\mathring{\nabla}_{\sigma}k^{\lambda})
\end{align}
where $\dot{k}^{\sigma}\coloneqq k^\mu \mathring{\nabla}_\mu k^\sigma$, which vanishes in the geodetic case.

\subsection{Curvature, torsion and non-metricity for the  connection \eqref{eq: gen conn}}\label{app: RTQ general conn}

The curvature form \eqref{eq: gen conn} is given by,
\begin{align}
\dfR_{ab} &\nonumber =\mathring{\dfR}_{ab}+\mathring{\Dex}\mathcal{C}_{ab}\wedge\dfk +\mathring{\Dex}\mathcal{P}_{cab}\wedge\spatcofr^{c}+k_a k_b \dex\dfA+g_{ab}\dex\dfB\\
 & \quad+2\left(\mathcal{C}_{c(a}\dfk+\mathcal{P}_{dc(a}\spatcofr^{d}\right)k_{b)}k^{c}\wedge\dfA +\mathcal{P}_{dcb}\mathcal{P}_{ea}{}^{c}\spatcofr^d\wedge\spatcofr^e-2\mathcal{P}_{d[a}{}^{c}\mathcal{C}_{b]c}\dfk\wedge\spatcofr^{d}\,,
\end{align}
and the torsion by 
\begin{align}
\dfT^a & =\mathcal{C}_{c}{}^{a}\dfk\wedge\cofr^c +\mathcal{P}_{cd}{}^{a}\cofr^{cd}+ k^a \dfA\wedge\dfk+\dfB\wedge\cofr^a \,,\\
 & =\left[-Ck^a -\kC^{a}-\kA k^a +B k^a -\kB l^{a}\right]\dfk\wedge\dfl\nonumber \\
 & \quad+\left[\kC_c l^a+C_{c}k^a +\sC_{c}{}^{a}-P_{c}l^{a}+P_{c}{}^{a}-k^a  \sA_c+B\delta_c^a- \sB_c l^a\right]\dfk\wedge\spatcofr^c \nonumber \\
 & \quad+\left[P_{c}k^a +\kP_{c}{}^{a}+\kB\delta_{c}^{a}- \sB_c k^a\right]\dfl\wedge\spatcofr^c \nonumber \\
 & \quad+\left[\kP_{cd}l^{a}+P_{cd}k^a + \sP_{cd}{}^{a}+ \sB_c\delta_d^a\right] \spatcofr^c\wedge\spatcofr^d\,,
\end{align}
with trace and antisymmetric components
\begin{align}
\dfT^{(\mathrm{tr})}{}^a & =\tfrac{1}{\dimM-1}(C+\kA-P_c{}^c)\cofr^a \wedge\dfk -\tfrac{1}{\dimM-1}\bar{P}_{b}{}^{b}\cofr^a \wedge\dfl\nonumber \\
 & \qquad+\tfrac{1}{\dimM-1}(\kC_c + \sP_{bc}{}^b)\cofr^a\wedge\cofr^c -\cofr^a \wedge\dfB\,,\\
\dfT^{(\mathrm{a})}{}^{a} & =g^{ab}\left[2(P_{[b}-\bar{C}_{[b})k_c l_{d]}+ (\sC_{[cd}+2P_{[cd})k_{b]}+2\bar{P}_{[cd} l_{b]}+ \sP_{[bcd]}\right]\cofr^c\wedge\cofr^d\,,
\end{align}
while the other one can be calculated simply by $\dfT^{(\mathrm{tn})}{}^{a} =\dfT^a-\dfT^{(\mathrm{tr})}{}^{a}-\dfT^{(\mathrm{a})}{}^{a}$. 

For the non-metricity we have the following expression
\begin{align}
\dfQ_{ab} & =2\dfom_{(ab)}=2k_a k_b \dfA+2g_{ab}\dfB\,.
\end{align}
Therefore the traces are
\begin{align}
\dfQ_{c}{}^{c} & =2D\dfB\,,\\
\dint{\vfre^{c}}\dfQ_{cb} & =2\kA k_b +2B_{b}\,,
\end{align}
and its irreducible decomposition,
\begin{align}
\dfQ^{(\mathrm{tr}1)}{}_{ab} & =2g_{ab}\dfB\,,\\
\dfQ^{(\mathrm{tr}2)}{}_{ab} & =\tfrac{4\dimM}{(\dimM-1)(\dimM+2)}\kA\left[k_{(a}\cofr_{b)}- \tfrac{1}{\dimM}g_{ab}\dfk\right]\,,\\
\dfQ^{(\mathrm{s})}{}_{ab} & = 2\big[k_{(a}k_b  A_{c)}- \tfrac{2}{\dimM+2}\kA k_{(a}g_{bc)}\big]\cofr^c\,,\\
\dfQ^{(\mathrm{tn})}{}_{ab} & = 2k_a k_b \dfA - \dfQ^{(\mathrm{tr}2)}{}_{ab} -\dfQ^{(\mathrm{s})}{}_{ab} \,.
\end{align}

\subsection{Irreducible decomposition of the curvature and the torsion for \eqref{eq: conn 2}}\label{app: RTQ conn 2}

The irreducible components of the torsion \eqref{eq: T for type II} are
\begin{align}
\dfT^{(\mathrm{tr})}{}^{a} & =\tfrac{1}{\dimM-1}(C+k_c  A^c )\cofr^a \wedge\dfk+\tfrac{1}{\dimM-1}\kC_c \cofr^a \wedge\spatcofr^c +\dfB\wedge\cofr^a \,,\nonumber \\
\dfT^{(\mathrm{a})}{}^{a} & =g^{ab}\left[-2\kC_{[b}k_c l_{d]}+C_{[cd}k_{b]}\right]\cofr^c\wedge\cofr^d\,,\\
\dfT^{(\mathrm{tn})}{}^{a} & =-\left[\tfrac{1}{3}\kC^{a}+\tfrac{\dimM-2}{\dimM-1}(C+k_c  A^c )k^a \right]\dfk\wedge\dfl+\tfrac{2\dimM-5}{3(\dimM-1)}\kC_c k^a \dfl\wedge\spatcofr^c \nonumber \\
 & \quad-\left(\tfrac{1}{\dimM-1}\kC_{d}\delta_{c}^{a}+\tfrac{1}{3}C_{cd}k^a \right)\spatcofr{}^{cd}\\
 & \quad+\left[\tfrac{1}{3}C_{c}{}^{a}+(C_{c}-\tilde{ A}_{c})k^a +\tfrac{\dimM-4}{3(\dimM-1)}\kC_c l^{a}+\tfrac{1}{\dimM-1}(C+k_d  A^{d})\delta_{c}^{a}\right]\dfk\wedge\spatcofr^c\,.
\end{align}
For the curvature \eqref{eq: R for type II} we first separate into antisymmetric and symmetric parts
\begin{align}
(\dfR_{[ab]}\equiv) \quad \dfW_{ab}  & =\mathring{\dfR}_{ab}+\mathring{\Dex}\mathcal{C}_{ab}\wedge\dfk\,,\\
(\dfR_{(ab)}\equiv) \quad \dfZ_{ab} & =k_a k_b \dex\dfA+g_{ab}\dex\dfB-2k_c k_{(a}\mathcal{C}_{b)}{}^{c}\dfk\wedge\dfA\,.
\end{align}
Taking this into account, it can be shown that the irreducible components are
\begin{align}
\dfW^{(3)}{}_{ab} & =\Big(\partial_{v}\sC_{[ab}l_c k_{d]}+\tilde{\partial}_{[c}\sC_{ab}k_{d]}\Big)\cofr^c\wedge\cofr^d\,,\\
\dfW^{(4)}{}_{ab} & =\mathring{\dfR}^{(4)}{}_{ab}-2\tfrac{\dimM-1}{\dimM-2}\dfW^{(6)}{}_{ab}\nonumber \\
   & \quad+\tfrac{1}{\dimM-2}\Big[\partial_{v}(2Cl_{[a}-C_{[a}) +\tilde{\partial}_{c}(2C^{c}k_{[a}+\kC^{c}l_{[a})-\tilde{\partial}_{c}(\sC_{[a}{}^{c}- C\delta_{[a}^{c})\Big]\dfk\wedge\cofr{}_{b]}\nonumber \\
   & \quad+\tfrac{1}{\dimM-2}\Big[\partial_{v}(2Ck_{[a}+\kC_{[a})+ \tilde{\partial}_{c} \kC^{c}k_{[a}\Big]\dfl\wedge\cofr{}_{b]}\nonumber \\
   & \quad+\tfrac{1}{\dimM-2}\Big[-\partial_{v}(C_{c}k_{[a}- \kC_c l_{[a})+ 2\tilde{\partial}_{(c}\kC_{d)}\delta_{[a}^{d}- \tilde{\partial}_{d}(\sC_{c}{}^{d}- C\delta_{c}^{d})k_{[a}\Big]\spatcofr^{c}\wedge\cofr{}_{b]}\,,\\
\dfW^{(5)}{}_{ab} & =\quad\tfrac{1}{\dimM-2}\Big[-\partial_{v}C_{[a}+ \tilde{\partial}_{c} \kC^{c}l_{[a}- \tilde{\partial}_{c}(\sC_{[a}{}^{c}+ C\delta_{[a}^{c})\Big]\dfk\wedge\cofr{}_{b]}\nonumber \\
   & \quad+\tfrac{1}{\dimM-2}\Big[\partial_{v}\kC_{[a}- \tilde{\partial}_{c} \kC^{c}k_{[a}\Big]\dfl\wedge\cofr{}_{b]}\nonumber \\
   & \quad+\tfrac{1}{\dimM-2}\Big[\partial_{v}(C_{c}k_{[a}- \kC_c l_{[a})+2\tilde{\partial}_{[c}\kC_{d]}\delta_{[a}^{d}+ \tilde{\partial}_{d}(\sC_{c}{}^{d}+C\delta_{c}^{d})k_{[a}\Big]\spatcofr^{c}\wedge\cofr{}_{b]}\,,\\
\dfW^{(6)}{}_{ab} & =\tfrac{2}{\dimM(\dimM-1)}(\partial_{v}C +\tilde{\partial}_{c} \kC^{c})\cofr_a\wedge\cofr_b\,,\\
\dfZ^{(2)}{}_{ab} & =\tfrac{1}{2(\dimM-2)}Z_c^{-}\dint{\vfre_{(a|}}\Big\{\dfk\wedge\spatcofr^{c}\wedge[\cofr_{|b)}-(\dimM-2)k_{|b)}\dfl]\Big\}\nonumber \\
   & \quad+\tfrac{1}{2}(e^{i}{}_{c}e^{j}{}_{d}\partial_{[i} A_{j]}+\kC_c \sA_{d})k_{(a}\dint{\vfre_{b)}}\big(\dfk\wedge\spatcofr^c\wedge\spatcofr^d\big)\,,\\
\dfZ^{(3)}{}_{ab}  & =\tfrac{\dimM}{\dimM^{2}-4}Z_c^{-}\left[k_{(a}\cofr_{b)}\wedge\spatcofr^c- \delta_{(a}^{c}\cofr_{b)}\wedge\dfk-\tfrac{2}{\dimM}g_{ab}\dfk\wedge\spatcofr^c\right]\,,\\
\dfZ^{(4)}{}_{ab} & =\tfrac{1}{\dimM}g_{ab}\dfZ_{c}{}^{c}\ =g_{ab}\dex\dfB\,,\\
\dfZ^{(5)}{}_{ab} & =\tfrac{1}{\dimM}Z_c^{+}k_{(a}\cofr_{b)} \wedge\spatcofr^{c} +\tfrac{1}{\dimM}\left[2\big(2(C\kA-\partial_{[u} A_{v]}) +\kC_c \sA^c\big)k_{(a}+ Z_{(a}^{+}\right]\cofr_{b)}\wedge\dfk\,,
\end{align}
where we have introduced the abbreviation $Z_a^\pm\coloneqq 2e^{i}{}_{a}\partial_{[v} A_{i]}\pm \kC_{a} \kA$, and the other three have been omitted because they can be calculated by the ones above by the use of the relations
\begin{align}
\dfW^{(2)}{}_{ab} & =\tfrac{1}{2}\dfW_{ab}+\tfrac{1}{4}\left(\dint{\vfre_{a}}\dint{\vfre_{b}}\dfW_{dc}\right)\cofr^d\wedge\cofr^c-\dfW^{(5)}{}_{ab}\,,\\
\dfW^{(1)}{}_{ab} & =\dfW_{ab}-\dfW^{(2)}{}_{ab}-\dfW^{(3)}{}_{ab}-\dfW^{(4)}{}_{ab}-\dfW^{(5)}{}_{ab}-\dfW^{(6)}{}_{ab}\,,\\
\dfZ^{(1)}{}_{ab} & =\dfZ_{ab}-\dfZ^{(2)}{}_{ab}-\dfZ^{(3)}{}_{ab}-\dfZ^{(4)}{}_{ab}-\dfZ^{(5)}{}_{ab}\,.
\end{align}

\subsection{Other expressions derived from the connection \eqref{eq: conn 2}}\label{app: other properties conn 2}

The general derivatives of $\dfk$ and $\dfl$ are
\begin{align}
\nabla_{c}k^a  & =-(Ck^a +\kC^{a})k_c +k^a  B_{c}\,,\\
(\nabla_{c}-\mathring{\nabla}_c)l^{a} & =(Cl^{a}-C^{a})k_c +k^a  A_{c}+l^{a} B_{c}\,.
\end{align}
With these equations and the following properties of the distorsion tensor (defined as the difference between the connection and the Levi-Civita one)
\begin{align}
g^{ca}(\dfom_{ca}{}^b-\mathring{\dfom}_{ca}{}^b) & =\left(\kA-C\right)k^b+ B^b-\kC^b\,,\\
k^c(\dfom_{ca}{}^b-\mathring{\dfom}_{ca}{}^b) & =\kA k_a k^b+\kB\delta_a^b\,,\\
l^c(\dfom_{ca}{}^b-\mathring{\dfom}_{ca}{}^b) & =\mathcal{C}_a{}^b+A k_a k^b+B\delta_a^b\,.
\end{align}
one can prove for transversal tensors of arbitrary number of indices
\begin{align}
k^{c}\nabla_{c}\tilde{S}_{a...}{}^{b...} & =\underbrace{k^{c}\mathring{\nabla}_{c}}_{\partial_{v}}\tilde{S}_{a...}{}^{b...}+(n^{\text{up}}-n_{\text{down}})\kB \tilde{S}_{a...}{}^{b...}\,,\\
l^{c}\nabla_{c}\tilde{S}_{a...}{}^{b...} & =l^{c}\mathring{\nabla}_{c}\tilde{S}_{a...}{}^{b...}+(n^{\text{up}}-n_{\text{down}})B\tilde{S}_{a...}{}^{b...}\nonumber \\
 & \quad-(\sC_{a}{}^{d}-l_{a}\kC^d-k_a C^d) \tilde{S}_{d...}{}^{b...}- ...+(\sC_{d}{}^{b}+l^b \kC_{d}+ k^b C_{d})\tilde{S}_{a...}{}^{d...}+...\,,\\
k^a \nabla_{c}\tilde{S}_{ab...}{}^{d...} & =\kC^{a}k_c S_{ab...}{}^{d...}\,,\\
l^{a}\nabla_{c}\tilde{S}_{ab...}{}^{d...} & =(C^{a}k_c -\mathring{\nabla}_{c}l^{a})\tilde{S}_{ab...}{}^{d...}\,,\\
\nabla^{a}\tilde{S}_{ab...}{}^{c...} & =\mathring{\nabla}^{a}\tilde{S}_{ab...}{}^{c...}+\left[\kC^{a}+(n^{\text{up}}-n_{\text{down}}-1) \sB^{a}\right]\tilde{S}_{ab...}{}^{c...}\,,
\end{align}
where $n^{\text{up}}$ and $n_{\text{down}}$ are respectively the
number of indices up (contravariance) and down (covariance) of the
tensor $\tilde{S}_{a...}{}^{b...}$. These properties are extremely useful in order to eliminate or reduce derivatives that appear in the equations of motion of metric-affine theories.

%%%%%%%%%%%%%%%%%%%%%%%%%%%%%%%%%%%%%%%%%%%%%%%%%%%%%%%%%%

%\bibitem{McCrea} J. D. McCrea, 
%Class. Quant. Grav. {\bf 9} (1992), no. 2, 553--568.

%\bibliographystyle{ieeetr}
%\bibliography{AnsatzWaveMAG_bibliog} 

%\end{multicols}

\end{document}